\documentclass[letterpaper,12pt,onecolumn]{IEEEtran}
\usepackage{graphicx,url,psfrag}
\usepackage{cite}
\usepackage{epic,eepic,subfigure}

\title{Feedback Control of Negative-Imaginary Systems\\
{\Large
  Flexible structures with colocated actuators and sensors}
}
\author{Ian R. Petersen and Alexander Lanzon --- \today}

\usepackage{latexsym}
\usepackage{amsmath}
\usepackage{amssymb}
\usepackage{amsfonts}
\usepackage{mathrsfs}
\usepackage{euscript}
\usepackage{pifont}
\usepackage{here}

\newcommand{\Real}{\ensuremath{{\mathbb{R}}}}

\newcommand{\Complex}{\ensuremath{{\mathbb{C}}}}

\newcommand{\jw}{\ensuremath{\jmath\omega}}

\newcommand{\eqn}[1]{\begin{equation}#1\end{equation}}
\newcommand{\matrixB}[1]{\begin{bmatrix}#1\end{bmatrix}}

\newtheorem{definition}{Definition}
\newtheorem{theorem}{Theorem}
\newtheorem{corollary}[theorem]{Corollary}
\newtheorem{lemma}[theorem]{Lemma}

\newcommand{\transpose}{\rm T}
\renewcommand{\jw}{\ensuremath{\jmath \omega}}
\newcommand{\Maeig}{\ensuremath{\lambda_{\mbox{\small max}}}}

\begin{document}
\maketitle

\section{}

Highly resonant dynamics can severely degrade the performance of
technological systems.  Structural modes in machines and robots,
ground and aerospace vehicles, and precision instrumentation, such as
atomic force microscopes and optical systems, can limit the ability of
control systems to achieve the desired performance.  Consequently,
control systems must be designed to suppress the effects of these
dynamics, or at least avoid exciting them beyond open-loop levels.
Open-loop techniques for highly resonant systems, such as input
shaping \cite{PRE02}, as well as closed-loop techniques, such as
damping augmentation \cite{BM90,DLSS08}, can be used for this
purpose.

Structural dynamics are often difficult to model with high precision
due to sensitivity to boundary conditions as well as aging and environmental effects.
Therefore, active damping augmentation to counteract the effects of external commands
and disturbances must account for parametric uncertainty and unmodeled dynamics.
This problem is simplified to some extent by using force actuators combined with
colocated measurements of velocity, position, or acceleration, where colocated
refers to the fact that the sensors and actuators have the same
location and the same direction.
Colocated control with velocity measurements, called {\it negative-velocity feedback},
can be used to directly increase the effective damping, thereby facilitating the design of
controllers that guarantee closed-loop stability in the
presence of plant parameter variations and unmodeled dynamics \cite{BAL79,PRE02}.
This guaranteed stability property can be
established by using  results on passive systems
\cite{DV75,KHA01}. However, the theoretical properties of
negative-velocity feedback are based on the idealized assumption of
colocation and require the availability of
velocity sensors, which may be expensive.
Also,  the choice of measured variable may depend on whether the desired
 objective is shape control or damping augmentation.

An alternative approach to negative-velocity feedback is {\it positive-position feedback},
where position sensors are used in place of velocity sensors.  Although position sensors can facilitate the objective of shape control,
it is less obvious how they can be used for damping augmentation.  Nevertheless,  it is shown in \cite{GC85,FC90} that
a positive-position feedback controller can be designed to increase the damping of the modes of a flexible structure.
Furthermore, this controller is robust against
uncertainty in the modal frequencies as well as
unmodeled plant dynamics.
As shown in \cite{GC85,FC90,LP06a,LP06}, the robustness properties of positive-position feedback
are similar to those of
negative-velocity feedback.

The present article
investigates the robustness of positive-position feedback control of
flexible structures with colocated force actuators and position
sensors. In particular, the theory of negative-imaginary
systems \cite{LP06a,LP06} is used to reveal the robustness properties of multi-input,
multi-output (MIMO) positive-position feedback controllers and related types of controllers for flexible
structures \cite{HM01,PMS02,PRE02,AFM07,BMP1a}.
The negative-imaginary property of linear systems can be extended to nonlinear systems through the notion of  counterclockwise input-output dynamics
\cite{ANG04,ANG06}. It is shown in \cite{POB05} for
the single-input, single-output (SISO) linear case  that  the results of \cite{ANG04,ANG06} guarantee the stability  of a
positive-position feedback control system in the presence of  unmodeled dynamics and parameter
uncertainties that maintain the negative-imaginary property of the plant.

Positive-position feedback  can
be regarded as one of the last areas of classical control theory to be encompassed by modern control theory. In this article,
positive-position feedback, negative-imaginary systems, and related
control methodologies are brought together with the underlying systems
theory.

 Table \ref{notation} summarizes   notation used
 in this article, while Table \ref{acronyms} lists acronyms.

\section{Flexible Structure Modeling}
In modeling an undamped flexible structure with a
single  actuator and a single sensor, modal analysis
 can be applied to  the relevant partial differential equation
 \cite{meirovitch86},
 leading to the
transfer function
\begin{equation}
\label{siso_tf_general}
P(s) = \sum_{i=1}^\infty \frac{\phi_i(s)}{s^2
  +\omega_i^2},
\end{equation}
where each  $\omega_i > 0$ is a  modal frequency, the functions
$\phi_i(s)$ are first-order polynomials, and $\omega_i\neq \omega_j$
for $i\neq j$.
In the case of a structure with a
 force actuator and  colocated velocity sensor, the form of the numerator of (\ref{siso_tf_general}) is
determined by
the passive nature of the flexible structure. Since the product $u(t)y(t)$ of the force
actuator input $u(t)$ and the velocity sensor output $y(t)$ represents
the power  provided by the actuator to the
structure at time
$t$, conservation of energy
 implies
\begin{equation}
\label{siso_dissipation}
E(t) \leq E(0) + \int_0^t u(\tau)y(\tau) d\tau
\end{equation}
for all $t \geq 0$, where $E(t)\geq 0$ represents the energy stored in the
system at time $t$, and $E(0)$ represents the initial energy stored in the
system.   In this case, the variables $u(t)$ and $y(t)$ are {\em dual}. The passivity condition (\ref{siso_dissipation})
  implies that the  transfer
 function $P(s)$ is positive real according to the following
 definition  \cite{DV75}.

\begin{definition}(\cite{BLME07,AV73})
\label{D1}
The square transfer function matrix $P(s)$ is {\em positive real} if the
following conditions are satisfied:
\begin{enumerate}
\item
All of the poles of $P(s)$ lie in CLHP.
\item
For all $s$ in  ORHP,
\begin{equation}
\label{pr0}
P(s)+P^*(s) \geq 0.
\end{equation}
\end{enumerate}
\end{definition}

If $P(s)$ is positive real, then it follows that \cite{BLME07,AV73}
\begin{equation}
\label{siso_PR}
P(\jw) + P^*(\jw) \geq 0
\end{equation}
for all $\omega\in\mathbb{R}$ such that $s=\jw$ is not a pole of
$P(s)$. If $P(s)$ is a SISO transfer function, then, for all
$\omega\in\mathbb{R}$ such that $s=\jw$ is neither a pole nor a zero of
$P(s)$, (\ref{siso_PR}) is equivalent to the phase condition
$\angle
P(\jw) \in [-\frac{\pi}{2},\frac{\pi}{2}]$.

\begin{definition} (\cite{BLME07})
\label{D1a}
The nonzero square transfer function matrix $P(s)$ is {\em strictly positive real} if
there exists  $\varepsilon > 0$ such that the transfer function matrix $P(s-\varepsilon)$ is positive real.
\end{definition}

If $P(s)$ is strictly positive real, then it follows \cite{BLME07} that
all of the poles of $P(s)$ lie in OLHP and
\begin{equation}
\label{spr0}
P(\jw)+P^*(\jw) > 0
\end{equation}
for all $\omega\in\mathbb{R}$.
If P(s) is a SISO transfer function, then (\ref{spr0}) holds for all $\omega\in \mathbb{R}$ such that $s=\jw$ is neither a pole nor a zero of
$P(s)$ if and only if the phase condition $\angle
P(\jw) \in (-\frac{\pi}{2},\frac{\pi}{2})$ holds for all $\omega\in \mathbb{R}$ such that $s=\jw$ is neither a pole nor a zero of
$P(s)$.

Now consider the positive-real transfer function from force actuation to velocity measurement given by
\begin{equation}
\label{siso_tf_PR}
P(s) = \sum_{i=1}^\infty \frac{\psi_i^2 s}{s^2 + \kappa_i s
  +\omega_i^2},
\end{equation}
where,  for all $i$, $\kappa_i > 0$ is the viscous damping constant associated with the
$i$th mode and $\omega_i > 0$. The transfer function
(\ref{siso_tf_PR}) satisfies the phase condition
$\angle P(\jw) \in (-\frac{\pi}{2},\frac{\pi}{2})$ for all
$\omega > 0$. However, \eqref{siso_tf_PR} has a zero at the origin, and thus (\ref{spr0})
is not satisfied for $\omega =0$. Hence, \eqref{siso_tf_PR} is not strictly
positive real.

Now consider a lightly damped flexible structure with
$m$ colocated sensor and actuator pairs.  Let $u_1(t),
\ldots, u_m(t)$ denote the force actuator input signals, and let $y_1(t),
\ldots, y_m(t)$ denote the corresponding velocity sensor output signals. The actuator and
sensor in the $i$th colocated actuator and sensor pair are {\em dual} when the product
$u_i(t)y_i(t)$ is equal to the power provided to the structure
by the $i$th
actuator  at time $t$. Now, we let
\[
Y(s) = P(s) U(s),
\]
where
\[
U(s) = \left[\begin{array}{c} U_1(s)  \\ \vdots \\ U_m(s)
\end{array}\right],~~
Y(s) = \left[\begin{array}{c} Y_1(s)  \\ \vdots \\ Y_m(s)
\end{array}\right].
\]
For $i=1,2,\ldots,m$, $U_i(s)$ and $Y_i(s)$ are the Laplace transforms of  $u_i(t)$ and $y_i(t)$, respectively,
and $P(s)$ is the transfer function matrix of the
system.
Then $P(s)$ is positive real and has the form
\begin{equation}
\label{mimo_tf_PR}
P(s) = \sum_{i=1}^\infty \frac{s}{s^2 + \kappa_i s
  +\omega_i^2}\psi_i \psi_i^{\transpose} ,
\end{equation}
where,  for all $i$, $\kappa_i > 0$, $\omega_i > 0$, and $\psi_i$ is an $m \times 1$ vector.
A review of positive-real and passivity theory is given in ``What Is Positive-real and Passivity Theory?''

\section{Negative-Imaginary Systems}

Mechanical structures
with colocated force
actuators and position sensors  do not yield
positive-real systems because the product of force and position is not equal to the power provided by the actuator \cite{LP06a,LP06}.
In this case, the transfer function matrix from the force actuator
inputs  $u_1(t),
\ldots, u_m(t)$ to
the position sensor outputs $y_1(t),
\ldots, y_m(t)$ is of the form
\begin{equation}
\label{mimo_tf_NI}
P(s) = \sum_{i=1}^\infty \frac{1}{s^2 + \kappa_i s
  +\omega_i^2}\psi_i \psi_i^{\transpose},
\end{equation}
where, for all $i$, $\kappa_i > 0$, $\omega_i > 0$, and $\psi_i$ is an
$m \times 1$ vector.
Therefore, the {\em Hermitian-imaginary part} \[\Im_{\text{H}}[P(\jw)] = -\frac{1}{2}\jmath(P(\jw) - P^{*}(\jw))\] of the frequency response function matrix
$P(\jw)$
satisfies
\begin{equation}
\Im_{\text{H}}[P(\jw)]
= -\omega\sum_{i=1}^\infty \frac{
  \kappa_i}{\left(\omega_i^2-\omega^2\right)^2 + \omega^2
  \kappa_i^2}\psi_i \psi_i^{\transpose}
\label{mimo_ni}
\leq 0
\end{equation}
for all $\omega\geq 0.$
That is, the frequency response
function matrix for the transfer function matrix
(\ref{mimo_tf_NI}) has negative-semidefinite Hermitian-imaginary part for all
$\omega \geq  0$.  We thus refer to the
transfer function matrix $P(s)$ in (\ref{mimo_tf_NI})
as negative imaginary. A formal definition follows.

\begin{definition}
\label{D3}
The square transfer function matrix $P(s)$ is {\em negative-imaginary
  (NI)} if the
following conditions are satisfied:
\begin{enumerate}
\item
\label{D3.1}
All of the poles of $P(s)$ lie in OLHP.
\item For all $\omega \geq 0$,
\label{D3.2}
\begin{equation}
\label{NI}
\jmath[P(\jw)-P^{*}(\jw)] \geq 0.
\end{equation}
\end{enumerate}
A linear time-invariant system is NI if its transfer function matrix
is NI.
\end{definition}

A discussion of
negative-imaginary  transfer functions arising in
electrical circuits is given in ``Applications to Electrical Circuits.''

In the SISO case, a transfer function is negative imaginary if and only if it
has no poles in CRHP and its
phase  is
in the interval $[-\pi,~0]$ at all frequencies that do not correspond to imaginary-axis poles or zeros. Consequently, the positive-frequency Nyquist
plot of a SISO negative-imaginary transfer function  lies below
the real axis as shown in  Figure~\ref{one-NIFR-system}. Hence, a
negative-imaginary
transfer function can be viewed as a  positive-real transfer function rotated clockwise by $90$
deg in the Nyquist plane.

Velocity sensors can be used in negative-velocity feedback control, whereas  position
sensors can be used in positive-position feedback
\cite{PRE02,GC85,FC90,BMP1a,AFM07,HM01,PMS02}.  Indeed,
positive-real theory and negative-imaginary theory
\cite{LP06a,LP06} achieve internal stability by a process referred to as {\em phase
  stabilization}, since instability is avoided by ensuring appropriate restrictions on the phase of
the corresponding open-loop systems. {\em Gain stabilization}, which is based on
the small-gain theorem \cite{BLME07}, guarantees robust
stability when the magnitude of the  loop transfer function is less than unity at
all frequencies. As in positive-real
analysis, robust stability of negative-imaginary systems \cite{LP06a,LP06} does not
require the magnitude of the loop transfer function to be less than
unity  at
all frequencies to guarantee
stability.
In order to present results on the robust stability of
positive-position feedback and related
control schemes, we now define MIMO
 strictly
negative-imaginary systems.

\begin{definition}
\label{D4}
The square transfer function matrix $P(s)$ is {\em strictly
  negative-imaginary (SNI)} if the
following conditions are satisfied:
\begin{enumerate}
\item
\label{D4.1}
All of the poles of $P(s)$ lie in  OLHP.
\item For all $\omega>0$,
\label{D4.2}
\begin{equation}
\label{SNI}
\jmath[P(\jw)-P^{*}(\jw)] > 0.
\end{equation}
\end{enumerate}
A linear time-invariant system is SNI if its transfer function matrix
is SNI.
\end{definition}

\begin{lemma}
\label{L1}
If the
$m\times m$ transfer function matrix $P_1(s)$
is   NI, respectively, SNI, and the
$m\times m$ transfer function matrix $P_2(s)$ is NI, then
\begin{equation}
\label{sum_ni}
P(s) = P_1(s)+P_2(s)
\end{equation}
is   NI, respectively, SNI.
\end{lemma}
\begin{proof}
This result follows directly from  Definition \ref{D3} and Definition \ref{D4}.
\end{proof}

\begin{theorem}
\label{LL2}
Consider the NI transfer function matrices
  $M(s)$ and $N(s)$, and suppose that the positive-feedback interconnection
shown in Figure \ref{feedback-interconnection}
is internally stable. Then the corresponding $2m\times 2m$ closed-loop transfer function matrix
\begin{equation}
\label{T}
T(s) = \left[\begin{array}{cc} M(s)\left(I-N(s)M(s)\right)^{-1} &
M(s)\left(I-N(s)M(s)\right)^{-1}N(s) \\
N(s)\left(I-M(s)N(s)\right)^{-1}M(s) &
N(s)\left(I-M(s)N(s)\right)^{-1}
\end{array}\right]
\end{equation}
 is NI. Furthermore, if, in addition, either
$M(s)$ or $N(s)$ is  SNI, then (\ref{T}) is SNI.
\end{theorem}
\begin{proof}
The internal stability of the positive feedback interconnection shown in Figure \ref{feedback-interconnection} implies that $T(s)$ is asymptotically stable. Given  $\omega \geq 0$, $w_1\in\Complex^{m}$, and $w_2\in\Complex^{m}$, define
\[
\matrixB{y_1 \\ y_2} =T(\jw)\matrixB{w_1 \\ w_2}.
\]
Letting $u_1=w_1+y_2$ and $u_2=w_2+y_1$, it follows from the positive feedback
interconnection that $y_1=M(\jw)u_1$ and $y_2=N(\jw)u_2$. Furthermore, using the fact that $M(s)$ and $N(s)$ are NI, it follows that
\begin{eqnarray*}
\lefteqn{\jmath \matrixB{w_1^* & w_2^*}[T(\jw) - T^*(\jw)]\matrixB{w_1
    \\ w_2}} \\
&=& \jmath \matrixB{w_1^* & w_2^*}\matrixB{y_1 \\ y_2}
- \jmath \matrixB{y_1^* & y_2^*}\matrixB{w_1 \\ w_2} \\
&=& \jmath \matrixB{u_1^*-y_2^* & u_2^*-y_1^*}\matrixB{y_1 \\ y_2}
- \jmath \matrixB{y_1^* & y_2^*}\matrixB{u_1-y_2 \\ u_2-y_1} \\
&=& \jmath\left(u_1^*y_1+u_2^*y_2\right)
-\jmath\left(y_1^*u_1+y_2^*u_2\right) \\
&=& \jmath\left(u_1^*M(\jw)u_1-u_1^*M(\jw)^*u_1 \right)
+\jmath\left(u_2^*N(\jw)u_2-u_2^*N(\jw)^*u_2 \right)\\
&\geq &0.
\end{eqnarray*}
 Since $\omega \geq 0$,
$w_1\in\Complex^{m}$, and $w_2\in\Complex^{m}$ are arbitrary, it
follows that
\[
\jmath[T(\jw)-T(\jw)^*] \geq 0
\]
for all $\omega \geq 0$ and hence, $T(s)$ is NI. The SNI result follows using similar arguments.
\end{proof}

\begin{theorem}
\label{LL3}
Consider the $2m\times 2m$  NI transfer function matrices
\[
M(s) = \left[\begin{array}{cc}M_{11}(s) & M_{12}(s) \\
M_{21}(s) & M_{22}(s)
\end{array}\right],~~
N(s) = \left[\begin{array}{cc}N_{11}(s) & N_{12}(s) \\
N_{21}(s) & N_{22}(s)
\end{array}\right],
\]
and suppose that the feedback
  interconnection shown in Figure \ref{star-product}
is internally stable.
Then the corresponding $2m\times 2m$ closed-loop transfer function matrix
\begin{eqnarray}
\label{TR}
T(s) &=& \left[\begin{array}{c}
M_{11}(2)+M_{12}(s)\left(I-N_{11}(s)M_{22}(s)\right)^{-1}N_{11}(s)M_{21}(s)  \\
N_{21}(s)\left(I-M_{22}(s)N_{11}(s)\right)^{-1}M_{21}(s)
\end{array}\right. \nonumber \\
&& \hspace{4cm}\left.\begin{array}{c}
M_{12}(s)\left(I-N_{11}(s)M_{22}(s)\right)^{-1}N_{12}(s)\\
N_{22}(s)+N_{21}(s)\left(I-M_{22}(s)N_{11}(s)\right)^{-1}M_{22}(s)N_{12}(s)
\end{array}\right]\nonumber \\
\end{eqnarray}
 is NI. Furthermore, if in addition, either
$M(s)$ or $N(s)$ is  SNI, then (\ref{TR}) is SNI.
\end{theorem}
\begin{proof}
The internal stability of the feedback interconnection shown in Figure \ref{star-product}
implies that $T(s)$ is asymptotically stable. Given  $\omega \geq
0$, $w_1\in\Complex^{m}$, and $w_2\in\Complex^{m}$, define
\[
\matrixB{y_1 \\ y_2} =T(\jw)\matrixB{w_1 \\ w_2}.
\]
Letting
\[
\matrixB{u_1 \\ u_2} = \left[\begin{array}{cc}
\left(I-N_{11}(s)M_{22}(s)\right)^{-1}N_{11}(s)M_{21}(s) &
\left(I-N_{11}(s)M_{22}(s)\right)^{-1}N_{12}(s) \\
\left(I-M_{22}(s)N_{11}(s)\right)^{-1}M_{21}(s) &
\left(I-M_{22}(s)N_{11}(s)\right)^{-1}M_{22}(s)N_{12}(s)
\end{array}\right]
\matrixB{w_1 \\ w_2},
\]
it follows from the feedback
interconnection shown in Figure \ref{star-product} that
\begin{equation}
\label{MN}
\matrixB{y_1 \\ u_2} = M(\jw) \matrixB{w_1 \\ u_1},~~
\matrixB{u_1 \\ y_2} = N(\jw) \matrixB{u_2 \\ w_2}.
\end{equation}
Furthermore, using (\ref{MN}) and the fact that $M(s)$ and $N(s)$ are NI, it follows that
\begin{eqnarray*}
\lefteqn{\jmath \matrixB{w_1^* & w_2^*}[T(\jw) - T^*(\jw)]\matrixB{w_1
    \\ w_2}} \\
&=& \jmath \matrixB{w_1^* & w_2^*}\matrixB{y_1 \\ y_2}
- \jmath \matrixB{y_1^* & y_2^*}\matrixB{w_1 \\ w_2} \\
&=& \jmath \left(\matrixB{w_1^* & u_1^*}\matrixB{y_1 \\ u_2}
-  \matrixB{y_1^* & u_2^*}\matrixB{w_1 \\ u_1}\right)
+ \jmath \left(\matrixB{u_2^* & w_2^*}\matrixB{u_1 \\ y_2}
-  \matrixB{u_1^* & y_2^*}\matrixB{u_2 \\ w_2}\right)\\
&=& \jmath \left(\matrixB{w_1^* & u_1^*}M(\jw)\matrixB{w_1 \\ u_1}
-  \matrixB{w_1^* & u_1^*}M(\jw)^*\matrixB{w_1 \\ u_1}\right) \\
&&+ \jmath \left(\matrixB{u_2^* & w_2^*}N(\jw)\matrixB{u_2 \\ w_2}
-  \matrixB{u_2^* & w_2^*}N(\jw)^*\matrixB{u_2 \\ w_2}\right)\\
&\geq &0.
\end{eqnarray*}
Since $\omega \geq 0$,
$w_1\in\Complex^{m}$, and $w_2\in\Complex^{m}$ are arbitrary, it
follows that
\[
\jmath[T(\jw)-T(\jw)^*] \geq 0
\]
for all $\omega \geq 0$ and hence, $T(s)$ is NI.
The SNI result follows using similar arguments.
\end{proof}

Underlying the stability properties of positive-position feedback
is the observation that the transfer
function matrix  of a lightly damped flexible
structure with colocated
force actuators and position sensors is NI.
Indeed, note that all poles of
\[
P_i(s) =  \frac{1}{s^2 + \kappa_i s
  +\omega_i^2}\psi_i \psi_i^{\transpose}
\]
in the transfer function matrix (\ref{mimo_tf_NI}) lie in  OLHP. Also,
for all $\omega \geq 0$,
\[
\jmath[P_i(\jw)-P_i^{*}(\jw)] =\Im_{\text{H}}(P_i(\jw))= \frac{2 \kappa_i \omega }{\left(\omega_i^2 - \omega^2\right)^2 + \kappa_i^2\omega^2}\psi_i \psi_i^{\transpose}
\geq 0.
\]
Hence, it follows from Definition \ref{D3} that each $P_i(s)$ is NI. Therefore, it
follows from  Lemma \ref{L1} that the transfer function matrix (\ref{mimo_tf_NI}) is NI.

\subsection{The Negative-Imaginary Lemma}
The following theorem, which is proved in \cite{LP06,XiPL1a}, provides a
state-space characterization of NI systems in terms of a pair of
linear matrix inequalities (LMIs). This
result is analogous to the positive-real lemma \cite{AV73,BLME07}, and
thus is  referred to as the {\em negative-imaginary lemma}.

\begin{theorem}
\label{NIL}
Consider the minimal state-space system
\begin{eqnarray}
\label{nilss1}
&\dot x = Ax + Bu,&  \\
\label{nilss2}
&y = Cx + Du,&
\end{eqnarray}
where $A \in \mathbb{R}^{n\times n}$, $B \in \mathbb{R}^{n\times m}$,
$C \in \mathbb{R}^{m\times n}$, and $D \in \mathbb{R}^{m\times
  m}$.   The
system (\ref{nilss1}), (\ref{nilss2}) is NI if and only if
$A$ has no eigenvalues on the imaginary axis,
$D$ is symmetric, and
  there exists a  positive-definite  matrix $Y \in \Real^{n\times n}$
  satisfying
\begin{equation}
\label{nil1}
AY+YA^{\transpose}\leq 0,
\end{equation}
\begin{equation}
\label{nil2}
B+AYC^{\transpose}=0.
\end{equation}
\end{theorem}

In Theorem \ref{NIL} it follows from the Lyapunov inequality~\eqref{nil1}, the positive definiteness of $Y$, and the assumption that $A$
has no eigenvalues on the imaginary axis that the matrix $A$ is asymptotically stable \cite[Corollary 11.8.1]{BER05}.

\begin{corollary}
\label{SNIL0}
Consider the minimal state-space system  (\ref{nilss1}),
(\ref{nilss2}),
where $A \in \mathbb{R}^{n\times n}$, $B \in \mathbb{R}^{n\times m}$,
$C \in \mathbb{R}^{m\times n}$, and $D \in \mathbb{R}^{m\times
  m}$.   The
system (\ref{nilss1}), (\ref{nilss2}) is SNI if and only if the
following conditions are satisfied:
\begin{enumerate}
\item
 $A$ has no eigenvalues on the imaginary axis.
\item
$D$ is symmetric.
\item
There exists a  positive-definite matrix $Y \in \mathbb{R}^{n\times
  n}$ such that
(\ref{nil1}) and (\ref{nil2}) are satisfied.
\item
The transfer function matrix $M(s) = C(sI-A)^{-1}B+D$ is such that
$M(s)-M^{\transpose}(-s)$ has no transmission
zeros on the imaginary axis except possibly at $s=0$.
\end{enumerate}
\end{corollary}

\begin{proof}
Assuming conditions 1)
- 3),  it follows from Theorem \ref{NIL} that
(\ref{nilss1}), (\ref{nilss2})  is NI. Now suppose
that  (\ref{nilss1}), (\ref{nilss2}) is not SNI. Then using Definition \ref{D3} and
Definition \ref{D4}, it follows that
there exist  $\omega > 0$ and a nonzero vector
$u \in \Complex^m$ such that
\[
\jmath u^*[M(\jw)-M^{*}(\jw)]u = 0.
\]
Thus, $M(s)-M^{\transpose}(-s)$ has a transmission zero at $s=\jw$, which
contradicts condition 4). Hence  (\ref{nilss1}),
(\ref{nilss2}) is SNI.

Conversely, suppose that  (\ref{nilss1}), (\ref{nilss2})
is SNI. Then,  (\ref{nilss1}), (\ref{nilss2}) is NI and
Theorem \ref{NIL} implies that conditions 1)
- 3) are satisfied. Also, it follows from Definition \ref{D4} that
\[
\jmath [M(\jw)-M^{*}(\jw)] > 0
\]
for all $\omega >0$. Therefore
$M(s)-M^{\transpose}(-s)$ has no transmission zeros on the imaginary
axis except possibly at $s=0$,
and thus condition 4) is satisfied.
\end{proof}

To illustrate Theorem \ref{NIL} and Corollary
\ref{SNIL0}, consider  the system
\begin{eqnarray}
\label{sys2}
&\dot x = -x + u,&
\\
 \label{sys2a}
 &y = x&
\end{eqnarray}
with transfer function
\begin{equation}
\label{first_order}
M(s) = \frac{1}{s+1}.
\end{equation}
The positive-frequency Nyquist plot of (\ref{first_order}) given in Figure \ref{F4.1}
shows that  (\ref{sys2}), (\ref{sys2a})  is both SNI and strictly positive real.

Applying Theorem \ref{NIL}  with  $A=-1$, $B=1$, $C=1$,
and $D=0$,  condition \eqref{nil2} can be satisfied by choosing
$Y = -\frac{B}{AC} = 1 > 0$. Then,  $AY+YA^{\transpose} = -2 < 0$. It now
follows from Theorem \ref{NIL} that  (\ref{sys2}), (\ref{sys2a}) is
NI. Also, note that
\[
M(s)-M^{\transpose}(-s) =  \frac{1}{s+1} - \frac{1}{-s+1} = \frac{2s}{s^2-1}
\]
 has no zeros on the imaginary axis except at $s=0$. It then follows
 from Corollary \ref{SNIL0} that
 (\ref{sys2}), (\ref{sys2a}) is SNI.

Now consider the transfer function
\begin{equation}
\label{second_order}
M(s) =
\frac{2s^2+s+1}{(s^2+2s+5)(s+1)(2s+1)}.
\end{equation}
The positive-frequency Nyquist plot of $M(\jw)$ in Figure
\ref{F4.2} shows that $\Im[M(\jw)] \leq 0$ for all $\omega \geq 0$, and thus $M(s)$ is  NI. However,
Figure \ref{F4.2} shows that there exists $\omega > 0$ such that $\Im[M(\jw)] = 0$, and thus $M(s)$ is not SNI.
Now consider the minimal realization  (\ref{nilss1}),
(\ref{nilss2}) of (\ref{second_order}) given by
\begin{eqnarray}
\label{sys3b}
&A = \left[\begin{array}{cccc}
   -3.5 &  -8.5  & -8.5 &  -2.5 \\
    1   &   0    &  0   &   0\\
    0   &   1    &  0   &   0\\
    0   &   0    &  1   &   0
\end{array}\right],~~
B=\left[\begin{array}{c}
    2.5\\
   -3\\
    1\\
    0
\end{array}\right],& \\
\label{sys3c}
&C= \left[\begin{array}{cccc}
    0&     0&     0&     1
\end{array}\right],~~D=0.&
\end{eqnarray}
In order to construct a matrix $Y$ satisfying the assumptions of
Theorem \ref{NIL}, note that the assumptions of
Theorem \ref{NIL} are equivalent to  the
requirement that the matrix $A$ have no eigenvalues on the imaginary
axis  and
\begin{eqnarray*}
&\left[\begin{array}{cc}
AY+YA^{\transpose} & B+AYC^{\transpose} \\
B^{\transpose} + C Y A^{\transpose} & 0
 \end{array}\right] \leq 0,&\\
& Y > 0.&
\end{eqnarray*}
Using  LMI software \cite{BCPS09}, we obtain
\[
Y =  \left[\begin{array}{cccc}
  100.375 &  -36.75 &    2.5 &    3 \\
  -36.75 &   18.5 &   -3 &   -1 \\
    2.5 &   -3 &    1 &        0 \\
    3 &   -1  &       0 &    0.2
\end{array}\right] > 0.
\]
Therefore Theorem \ref{NIL} implies that
(\ref{nilss1}),
(\ref{nilss2}), (\ref{sys3b}), (\ref{sys3c}) is NI.

Now to determine whether  (\ref{nilss1}),
(\ref{nilss2}),
(\ref{sys3b}), (\ref{sys3c}) is SNI, note that
\begin{eqnarray*}
&M(s)-M^{\transpose}(-s) =  \frac{2s^2+s+1}{(s^2+2s+5)(s+1)(2s+1)} -
 \frac{2s^2-s+1}{(s^2-2s+5)(-s+1)(-2s+1)}&\\
 &= \frac{-24(s^2+1)^2}{4s^8
   + 19s^6 +71 s^4 -119s^2 +25},&
\end{eqnarray*}
 has a double  zero at $s = \jmath$. Consequently,
(\ref{nilss1}),
(\ref{nilss2}), (\ref{sys3b}), (\ref{sys3c}) is not
SNI.

\subsection{Two Strict Negative-Imaginary Lemmas}
The following theorems  give  sufficient conditions for the SNI
property.

\begin{theorem}
\label{SNIL1}
 Consider the minimal state-space system  (\ref{nilss1}),
(\ref{nilss2}),
where $A \in \mathbb{R}^{n\times n}$, $B \in \mathbb{R}^{n\times m}$,
$C \in \mathbb{R}^{m\times n}$, and $D \in \mathbb{R}^{m\times
  m}$. Suppose the following
conditions are satisfied:
\begin{enumerate}
\item
All eigenvalues of $A$ are in  OLHP.
\item
 $D$ is symmetric.
\item
  There exist a  positive-definite matrix $\tilde Y \in
  \mathbb{R}^{n\times n}$ and positive numbers
  $\alpha, \varepsilon$  such that $-\alpha$ is not an
  eigenvalue of $A$ and the
  matrices
\begin{equation}
\label{SNIL1mat}
\tilde A = \left[\begin{array}{cc} A & 0\\0 & -\alpha I \end{array}\right], ~~
\tilde B =  \left[\begin{array}{c} B \\ \varepsilon I \end{array}\right], ~~
\tilde C = \left[\begin{array}{cc} C & - I  \end{array}\right]
\end{equation}
satisfy
\[
\tilde A \tilde Y+\tilde Y \tilde A^{\transpose}\leq 0
\]
and
\[
\tilde B+\tilde A \tilde Y \tilde C^{\transpose}=0.
\]
\end{enumerate}
Then  (\ref{nilss1}),
(\ref{nilss2})
is SNI.
\end{theorem}

The proof of Theorem \ref{SNIL1} requires the following lemma.

\begin{lemma}
\label{L2}
 Let $\varepsilon > 0$ and $\alpha > 0$. Then the transfer function matrix
\begin{equation}
\label{TFM}
M(s) = \frac{\varepsilon}{s+\alpha}I
\end{equation}
is SNI.
\end{lemma}

\begin{proof}
Let the transfer function matrix \eqref{TFM} have minimal state-space realization
\begin{eqnarray}
\label{sysalph}
&\dot x = -\alpha x + \varepsilon u,& \\
 \label{sysalpha}
 &y =  x.
\end{eqnarray}
Theorem \ref{NIL} and Corollary \ref{SNIL0} can be applied to (\ref{sysalph}), (\ref{sysalpha}) with $A = -\alpha I$, $B = \varepsilon I$, $C = I$, and $D =
0$. Setting $Y = \frac{\varepsilon}{ \alpha} I > 0$, it follows that $AY+YA^{\transpose} = -2 \varepsilon I < 0$ and
$B+AYC^{\transpose} = \varepsilon I-\frac{\alpha \varepsilon}{\alpha} I = 0$. Hence,
Theorem \ref{NIL} implies that  (\ref{sysalph}), (\ref{sysalpha}) is NI. Furthermore,
\begin{eqnarray*}
M(s)-M^{\transpose}(-s) &=& \frac{\varepsilon}{s+\alpha}I -
\frac{\varepsilon}{-s+\alpha}I  \nonumber \\
&=& \frac{2 \varepsilon s}{s^2 - \alpha^2}I.
\end{eqnarray*}
Thus, $M(s)-M^{\transpose}(-s)$ has no  purely imaginary
transmission zeros
except possibly at $s=0$. Hence, it follows from Corollary \ref{SNIL0} that
(\ref{sysalph}), (\ref{sysalpha}) is SNI.
\end{proof}

{\em Proof of Theorem \ref{SNIL1}:}
Let $\hat{M}(s)$ be the transfer function matrix of
(\ref{nilss1}),
(\ref{nilss2}).
Since $s=-\alpha$ is not a pole of $\hat{M}(s)$, a minimal state-space realization of the transfer function matrix
$M_1(s) = \hat{M}(s) - \frac{\varepsilon}{s+\alpha}I$ is
\begin{eqnarray*}
\label{sysauga}
&\dot x_1 = A x_1 + B u,&\\
\label{sysaugb}
&\dot x_2 = -\alpha x_2 + \varepsilon u,&\\
\label{sysaugc}
&y = Cx_1- x_2 + Du.&
\end{eqnarray*}
Let
\[
\tilde A = \left[\begin{array}{cc} A & 0\\0 & -\alpha I \end{array}\right], ~~
\tilde B =  \left[\begin{array}{c} B \\ \varepsilon I \end{array}\right], ~~
\tilde C = \left[\begin{array}{cc} C & - I  \end{array}\right], ~~
\tilde D = D.
\]
Assuming conditions 1) - 3), it follows from  Theorem \ref{NIL} that $M_1(s)$ is
NI.
Then  Lemma \ref{L1} and Lemma \ref{L2} imply
that  $\hat{M}(s) =
M_1(s) + \frac{\varepsilon}{s+\alpha}I$  is SNI.
\hfill \QED

To illustrate Theorem~\ref{SNIL1}, we consider lightly damped flexible
structures with force actuators and position sensors.
 An \emph{integral resonant controller}
 \cite{AFM07,BMP1a} has the form
\begin{equation}
\label{irc}
C(s) = [sI+\Gamma \Phi]^{-1}\Gamma,
\end{equation}
 where
$\Gamma$ and $\Phi$ are positive-definite matrices.
In the SISO case \cite{AFM07}, integral resonant controllers are
derived by first adding a direct feedthough to a resonant system with a
colocated force actuator and position sensor. Then, application of
integral feedback leads to damping of the resonant poles. Combining
the direct feedthrough with the integral feedback leads to a SISO
controller of the form (\ref{irc}). In \cite{BMP1a}, this class of
SISO controllers is generalized to MIMO controllers of the form
(\ref{irc}).

Integral resonant controllers provide integral force feedback
\cite{PRE02}, which refers
to control that uses position actuators, force sensors, and integral feedback. In \cite{PRE02}, integral
feedback is modified by  moving the integrator pole
slightly to the left in the complex plane to alleviate
actuator saturation.  A SISO controller transfer function
of the form (\ref{irc}) results from this process.

\begin{theorem}
\label{IRC_SNI}
The transfer function matrix (\ref{irc}) with $\Gamma$  positive definite and $\Phi$
positive definite is SNI.
\end{theorem}

{\em Proof:}
 Consider the
minimal state-space realization of (\ref{irc}) given by
\begin{eqnarray*}
&\dot x = -\Gamma \Phi x + \Gamma u,& \\
&y = x.&
\end{eqnarray*}
Let $\varepsilon > 0$ and $\alpha > 0$ be such that $-\alpha$ is not an
eigenvalue of $-\Gamma \Phi$. The corresponding matrices in (\ref{SNIL1mat}) are
\[
\tilde A = \left[\begin{array}{cc} -\Gamma \Phi & 0\\0 & -\alpha I
  \end{array}\right], ~~
\tilde B =  \left[\begin{array}{c} \Gamma \\ \varepsilon I \end{array}\right], ~~
\tilde C = \left[\begin{array}{cc} I & - I  \end{array}\right], ~~
\tilde D = 0.
\]
Also, let
\[
\tilde Y = \left[\begin{array}{cc}  \Phi^{-1} & 0\\0 & 0
  \end{array}\right]
+ \varepsilon \left[\begin{array}{cc}
\left(\frac{1}{\alpha}+1\right)I &  \left(\frac{1}{\alpha}+1\right)I\\
 \left(\frac{1}{\alpha}+1\right)I &  I
  \end{array}\right].
\]
Thus,
\begin{equation}
\label{tildaeq}
\tilde B+\tilde A \tilde Y
\tilde C^{\transpose}=0.
\end{equation}
Furthermore, note that
\[
\left[\begin{array}{cc}  \Phi^{-1} & 0\\0 & 0
  \end{array}\right]
\]
is positive semidefinite, and
\[
\left[\begin{array}{cc}
\left(\frac{1}{\alpha}+1\right)I &  \left(\frac{1}{\alpha}+1\right)I\\
 \left(\frac{1}{\alpha}+1\right)I &  I
  \end{array}\right]
\]
is positive definite. Hence,
$\tilde Y > 0$. 

Using the definitions of $\tilde A$ and $\tilde Y$, it follows
that
\[
\tilde A \tilde Y + \tilde Y \tilde A^{\transpose} =
 \left[\begin{array}{cc}  -\Gamma & 0\\0 & 0
  \end{array}\right] +
\varepsilon  \left[\begin{array}{cc}
-\left(\frac{1}{\alpha}+1\right)(\Gamma \Phi+\Phi \Gamma)
&-\left(\frac{1}{\alpha}+1\right)(\Gamma \Phi +\alpha I) \\
-\left(\frac{1}{\alpha}+1\right)(\Gamma \Phi +\alpha I)^{\transpose}
& -2  \left(\alpha+1\right) I
  \end{array}\right].
\]
Furthermore, the matrix
\[
\left[\begin{array}{cc}  \Gamma & 0\\0 & 0
  \end{array}\right]
\]
is positive semidefinite.
For every nonzero vector of the form $x = [0~~x_2^{\transpose}]^{\transpose}$,
we have
\[
\left[\begin{array}{c} 0\\x_2
  \end{array}\right]^{\transpose}\left[\begin{array}{cc}
\left(\frac{1}{\alpha}+1\right)(\Gamma \Phi+\Phi \Gamma)
&\left(\frac{1}{\alpha}+1\right)(\Gamma \Phi +\alpha I) \\
\left(\frac{1}{\alpha}+1\right)(\Gamma \Phi +\alpha I)^{\transpose}
& 2  \left(\alpha+1\right) I
  \end{array}\right]
\left[\begin{array}{c} 0\\x_2
  \end{array}\right] > 0.
\]
Hence, it follows using Finsler's theorem (see ``What Is Finsler's Theorem?''), Lemma \ref{Finsler}, that
there exists $\bar \tau > 0$ such
that
\[
\left[\begin{array}{cc}
\left(\frac{1}{\alpha}+1\right)(\Gamma \Phi+\Phi \Gamma)
&\left(\frac{1}{\alpha}+1\right)(\Gamma \Phi +\alpha I) \\
\left(\frac{1}{\alpha}+1\right)(\Gamma \Phi +\alpha I)^{\transpose}
& 2  \left(\alpha+1\right) I
  \end{array}\right]
+ \tau
\left[\begin{array}{cc}  \Gamma & 0\\0 & 0
  \end{array}\right] \geq 0
\]
for all $\tau \geq \bar \tau$. Let $\hat \varepsilon = \bar \tau ^{-1}> 0$. Consequently, choosing $\varepsilon \leq\hat \varepsilon$ implies
\begin{equation}
\label{tildelmi}
\tilde A \tilde Y + \tilde Y
\tilde A^{\transpose} \leq 0.
\end{equation}
 Combining (\ref{tildaeq}) and (\ref{tildelmi}), it follows that
conditions 1) - 3) of Theorem \ref{SNIL1} are satisfied, and therefore, the
transfer function (\ref{irc}) is SNI.
\hfill $\blacksquare$


\begin{theorem}
\label{SNIL2}
Consider the minimal state-space system (\ref{nilss1}),
(\ref{nilss2}),
where $A \in \mathbb{R}^{n\times n}$, $B \in \mathbb{R}^{n\times m}$,
$C \in \mathbb{R}^{m\times n}$, and $D \in \mathbb{R}^{m\times
  m}$. Suppose the following
conditions are satisfied:
\begin{enumerate}
\item
All of the eigenvalues of $A$ are in  OLHP.
\item
$D$ is symmetric.
\item
  There exist a positive-definite matrix $\tilde Y\in \mathbb{R}^{n\times
  n}$ and positive numbers
  $\varepsilon$, $\alpha$, and $\beta$ such that $\alpha \neq
  \beta$, $-\alpha,-\beta$ are not eigenvalues of $A$, and the
  matrices
\[
\tilde A = \left[\begin{array}{ccc} A & 0 & 0\\
0 & -\alpha I & 0\\
0 &0 & -\beta I
\end{array}\right], ~~
\tilde B =  \left[\begin{array}{c} B \\ \varepsilon I \\  \varepsilon I
\end{array}\right], ~~
\tilde C = \left[\begin{array}{ccc} C & - I & -I \end{array}\right]
\]
satisfy
\[
\tilde A \tilde Y+\tilde Y \tilde A^{\transpose}\leq 0
\]
and
\[
\tilde B+\tilde A \tilde Y \tilde C^{\transpose}=0.
\]
\end{enumerate}
Then
(\ref{nilss1}),
(\ref{nilss2})
is SNI.
\end{theorem}

The proof of Theorem \ref{SNIL1} requires the following lemmas.

\begin{lemma}
\label{L3}
Let $\varepsilon > 0$,  $\alpha > 0$, and $\beta > 0$. Then the  transfer function
\begin{equation}
\label{TF2}
M(s) = \frac{\varepsilon}{(s+\alpha)(s+\beta)}
\end{equation}
is SNI.
\end{lemma}

\begin{proof}
 The transfer function  \eqref{TF2} has a  minimal state-space realization
\begin{eqnarray}
\label{sysaug1a}
&\dot x
  = Ax +
Bu,&\\
\label{sysaug1b}
&y =C  x, &
\end{eqnarray}
where
\begin{eqnarray*}
A &=& \left[\begin{array}{cc} -\alpha  & 0 \\
0 &  -\beta \end{array}\right],~~ B=\left[\begin{array}{c}  \varepsilon
 \\  \varepsilon  \end{array}\right],~~
C= \left[\begin{array}{cc}   1   &   1 \end{array}\right].
\end{eqnarray*}
Applying Theorem \ref{NIL}
and Corollary \ref{SNIL0}  to
(\ref{sysaug1a}), (\ref{sysaug1b}),
and
setting
\[
Y = \left[\begin{array}{cc}
\frac{\varepsilon}{ \alpha} & 0 \\ 0 & \frac{\varepsilon}{ \beta}
\end{array}\right]
 > 0,
\]
 it follows that $AY+YA^{\transpose} = -2 \varepsilon I < 0$ and
$B+AYC^{\transpose} = 0$. Hence,  Theorem \ref{NIL} implies that
(\ref{sysaug1a}), (\ref{sysaug1b}) is NI. Furthermore, for  (\ref{sysaug1a}), (\ref{sysaug1b}),
$M(s)-M^{\transpose}(-s)$ is given by
\begin{eqnarray*}
M(s)-M^{\transpose}(-s) &=& \frac{\varepsilon}{(s+\alpha)(s+\beta)} -
\frac{\varepsilon}{(-s+\alpha)(-s+\beta)} \\
&=& \frac{2 \varepsilon (\alpha+\beta)s}{s^2(\alpha+\beta)^2 -
  (s^2+\alpha\beta)^2}.
\end{eqnarray*}
Since $M(s)-M^{\transpose}(-s)$ has no  imaginary transmission zeros
except at $s=0$, it follows from Corollary \ref{SNIL0} that
(\ref{sysaug1a}), (\ref{sysaug1b}) is SNI.
\end{proof}

\begin{lemma}
\label{L4}
If $M(s)$ is an SISO SNI transfer function, then the transfer function matrix $M(s) I$ is SNI.
\end{lemma}

\begin{proof}
This result follows directly from Definition \ref{D4}.
\end{proof}

\hspace{-0.5cm}{\em Proof of Theorem \ref{SNIL2}:}
Let  $\hat{M}(s)$ be the  transfer function matrix
of
(\ref{nilss1}),
(\ref{nilss2}).
Since neither $s=-\alpha$ nor $s=-\beta$ is a pole of $\hat{M}(s)$,
a minimal state-space realization  of
$M_1(s) = \hat{M}(s) - \frac{\varepsilon}{(s+\alpha)(s+\beta)}I$ is
\begin{eqnarray*}
&\dot x_1 = A x_1 + B u,&\\
&\dot x_2 = -\alpha x_2 + \varepsilon u,&\\
&\dot x_3 = -\beta x_3 + \varepsilon u,&\\
&y = Cx_1- x_2 -x_3+ Du.&
\end{eqnarray*}
Let
\[
\tilde A = \left[\begin{array}{ccc} A & 0 & 0\\
0 & -\alpha I & 0\\
0 &0 & -\beta I
\end{array}\right], ~~
\tilde B =  \left[\begin{array}{c} B \\ \varepsilon I \\  \varepsilon I
\end{array}\right], ~~
\tilde C = \left[\begin{array}{ccc} C & - I & -I \end{array}\right], ~~
\tilde D = D.
\]
Assuming conditions 1) - 3), it follows from Theorem \ref{NIL}
 that  $M_1(s)$
is NI. Finally,
 Lemma \ref{L1},  Lemma \ref{L3}, and Lemma \ref{L4}
imply that
$\hat{M}(s) =
M_1(s) + \frac{\varepsilon}{(s+\alpha)(s+\beta)}I$ is SNI.
\hfill\QED

\section{Robust Stability  of Negative-Imaginary Control   Systems}

We now present a result given by Theorem~\ref{thm:main_result} below that guarantees the
robustness and stability  of control systems involving the positive-feedback interconnection of an NI system
and an SNI system. This positive-feedback interconnection is illustrated in
Figure \ref{feedback-interconnection}.
The result is analogous  to the  passivity theorem given in ``What
Is Positive-real and Passivity Theory?''
concerning the negative-feedback interconnection of a positive-real system
and a strictly positive-real system.

Theorem~\ref{thm:main_result} guarantees the
internal stability of the positive-feedback interconnection of two systems
through phase stabilization, as opposed to
gain stabilization in the small-gain theorem.  In
phase stabilization the gains of the systems can be arbitrarily large, but
the phase of the loop transfer function  needs to be such that the
critical Nyquist point is not encircled by the Nyquist plot. In the  passivity theorem given in ``What
Is Positive-real and Passivity Theory?'', negative feedback is used, and thus the Nyquist point is
 at $s=-1+\jmath 0$. Then the cascade of  two positive-real
systems gives a loop transfer function whose phase is in the interval
$(-\pi,\pi)$. Hence, the Nyquist plot excludes the negative real axis. In
NI systems, positive feedback interconnection is
used and thus the Nyquist point is $s=1+\jmath 0$. This alternative Nyquist point is
required since an NI system has a phase lag in the interval $(-\pi,
0)$ and thus  two   NI systems in cascade have a phase
lag in the interval
$(-2\pi,0)$. That is, the Nyquist plot excludes the positive-real axis.

The following lemma is required in order to state the result given in Theorem \ref{thm:main_result} below.

\begin{lemma}
\label{L5}
Let $M(s)$ be an NI transfer function matrix. Then $M(\infty)$ and  $M(0)$ are symmetric, and
\begin{equation}
\label{ineq1}
  M(0)-M(\infty) \geq 0. 
\end{equation}
Also, let $N(s)$ be an  SNI transfer
  function matrix. Then $N(\infty)$ and $N(0)$ are symmetric, and
\begin{equation}
\label{ineq2}
  N(0)-N(\infty) >  0. 
\end{equation}
If, in addition, $N(\infty)$ is positive semidefinite, then $N(0)$ is positive definite and all
of the eigenvalues of the matrix $M(0)N(0)$ are
real.
\end{lemma}

\begin{proof}
See \cite{LP06}.
\end{proof}

\begin{theorem} \label{thm:main_result}
  Consider the NI transfer function matrix
  $M(s)$ and the SNI transfer
  function matrix $N(s)$, and suppose that $M(\infty)N(\infty)=0$
  and $N(\infty)\geq 0$. Then, the positive-feedback interconnection
  of $M(s)$ and $N(s)$
  is internally stable if and only if
\begin{equation}
\label{DC_gain}
  \Maeig(M(0)N(0))<1.
\end{equation}
\end{theorem}

\begin{proof}
See \cite{LP06}.
\end{proof}

In the MIMO case, the proof of
Theorem \ref{thm:main_result} given in \cite{LP06} uses Theorem
\ref{NIL}. In the SISO case, the sufficiency part of Theorem~\ref{thm:main_result} follows directly from  Nyquist
arguments and thus has an intuitive interpretation. For example, consider
\begin{equation}\label{M}
M(s) = \frac{1}{s+1},
\end{equation}
whose positive-frequency Nyquist plot is shown
in Figure \ref{F4.1}. Also consider
\begin{equation}\label{N}
N(s) =
\frac{2s^2+s+1}{(s^2+2s+5)(s+1)(2s+1)},
\end{equation}
whose positive-frequency Nyquist plot is
shown in Figure \ref{F4.2}. Figure \ref{F4.1} shows that $N(s)$ is SNI,
whereas Figure \ref{F4.2} shows that
$M(s)$ is NI but not SNI. The positive-frequency Nyquist plot of the
corresponding loop transfer function $L(s) = N(s)M(s)$ is shown
in Figure \ref{F5.2}.
 Since both
$N(s)$ and $M(s)$ have no poles in  CRHP,
and the Nyquist plot
of $L(s)$ does not encircle the critical point $s=1+\jmath 0$,   it follows
that the positive-feedback
interconnection of $M(s)$ and $N(s)$ is internally stable. A similar Nyquist argument  is
mentioned in \cite{FC90} as a justification for
the stability of SISO positive-position feedback systems.
Furthermore,  a condition equivalent to (\ref{DC_gain}) is
required in the result of \cite{ANG06}.

Consider $M(s)$ and $N(s)$ as in Theorem \ref{thm:main_result} in the SISO case. Since $N(s)$
is  SNI, it follows that
$\angle N(\jw) \in (-\pi,0)$ for all $\omega >0$. Furthermore, since $M(s)$ is  NI, it follows that  $\angle M(\jw) \in
[-\pi,0]$ for all $\omega \geq  0$ such that $M(\jw)\neq 0$. Hence,  $L(s) = M(s)N(s)$ satisfies $\angle L(\jw) \in
(-2\pi,0)$ for all $\omega >0$ such that $L(\jw)\neq 0$. Thus  the Nyquist plot
of $L(\jw)$ can  intersect the positive-real axis
 only at  $\omega = 0$ since at infinite frequency $M(\infty)N(\infty)=0$. Thus, the Nyquist plot
of $L(\jw)$ does not encircle the critical point $s=1+\jmath 0$ if
$M(0)N(0) < 1$. Hence, in the  SISO
case, the sufficiency part of Theorem \ref{thm:main_result} follows from the Nyquist
test.

A discussion on how rigid-body modes can be handled using Theorem \ref{thm:main_result} is given in ``How Are Rigid-Body Modes Handled?''.

\section{Negative-Imaginary  Feedback  Controllers}

We now apply  Theorem \ref{thm:main_result}
to NI feedback  control systems in the
case where one of the
blocks in the feedback connection shown in Figure
\ref{feedback-interconnection}
corresponds to the plant, while the other block corresponds to the
controller. This situation is shown in Figure \ref{F5.2a}.

Since flexible structures with
colocated force actuators and position sensors are typically SNI, Theorem \ref{thm:main_result} implies that  NI controllers guarantee closed-loop internal stability if the dc gain condition (\ref{DC_gain}) is satisfied. Indeed, many
schemes considered for controlling
flexible structures rely on controllers that are NI. These schemes include positive-position feedback
\cite{GC85,FC90,MVB06,PRE02}, resonant feedback
control
\cite{HM01,PMS02}, and integral resonant
control  \cite{AFM07,BMP1a}. We now consider each of these
control schemes in more detail.

\subsection{Positive-Position Feedback}
In the SISO case, a positive-position feedback
controller is a controller of the form
\begin{equation}
\label{siso-ppf}
C(s) = \sum_{i=1}^M\frac{k_i}{s^2 + 2\zeta_i\omega_{i} s +
  \omega_{i}^2},
\end{equation}
where $\omega_{i}>0$, $\zeta_i>0$, and $k_i>0$ for $i = 1,2,\ldots,M$.
Using Nyquist arguments, the SISO  transfer
function $C(s) = \frac{k}{s^2 + 2\zeta\omega s +
  \omega^2}$, where $\omega, \zeta, k>0$, is SNI. Consequently, it
follows from Lemma \ref{L1} that (\ref{siso-ppf}) is
SNI.
Furthermore,  this result can be extended to the
MIMO case to show that the transfer function
matrix
\begin{equation}
\label{mimo_ppf}
C(s) = K^{\transpose}(s^2I+Ds+\Omega)^{-1}K,
\end{equation}
where $D > 0$ and $\Omega > 0$, is SNI
\cite{LP06a}.
A MIMO positive-position feedback controller is a controller of the form (\ref{mimo_ppf}), while a  {\em positive-position feedback system} is a control system for a flexible structure with colocated
force actuators and position sensors with a controller of the
form (\ref{mimo_ppf})
\cite{GC85,FC90,MVB06,PRE02}.

The Nyquist proof of Theorem
\ref{thm:main_result}  justifies the use of positive-position feedback in
  the SISO case. That is, since the positive-position feedback
controller   (\ref{siso-ppf})
is  SNI, its phase is in the interval
$(-\pi,0)$ for all   $\omega >0$. Furthermore, since the flexible structure plant
is  NI, its phase is in the interval $[-\pi,0]$ for all
 $\omega \geq 0$ such that $\jw$ is not a zero. Hence,  the phase of the loop  transfer function is
in the interval $(-2\pi,0)$ for all   $\omega >0$ such that $\jw$ is not a zero. This fact, together with
the strict properness of the controller (\ref{siso-ppf}), implies
that the Nyquist plot
of the loop   transfer function can  intersect the positive-real
axis  at only
   the frequency $\omega =0$. Thus, the Nyquist plot
of the loop  transfer function does not encircle the critical point $s=1+\jmath 0$ if the
dc value of the loop transfer function  is strictly less than unity.

\subsection{Resonant Control}
We now consider the exactly proper SISO  SNI controller
\begin{equation}
\label{resonant1}
C(s) = \sum_{i=1}^M\frac{-k_is^2}{s^2 + 2\zeta_i\omega_{i} s +
  \omega_{i}^2},
\end{equation}
 where $\omega_{i} >0$, $\zeta_i>0$, and $k_i>0$ for $i =
 1,2,\ldots,M$. The controller (\ref{resonant1}) can be implemented as the
 positive-position feedback controller (\ref{siso-ppf}) using an acceleration
 sensor rather than a position sensor.  Alternatively, (\ref{resonant1}) can be
 implemented as the  positive-real feedback controller
\[
\bar C(s) = \sum_{i=1}^M\frac{k_is}{s^2 + 2\zeta_i\omega_{i} s +
  \omega_{i}^2},
\]
where $\omega_{i} >0$, $\zeta_i>0$, and $k_i>0$ for $i =
 1,2,\ldots,M$, using a velocity sensor rather than a position sensor.
To see that  (\ref{resonant1})  defines an NI
controller, we rewrite (\ref{resonant1}) as
$C(s) = -s^2\tilde C(s)$, where $\tilde C(s)$ is a
SISO  positive-position feedback
controller of the form defined in (\ref{siso-ppf}).  If
 $s=\jw$ and  $\omega > 0$, then $-s^2 = \omega^2 >
 0$. Therefore, since
 $\tilde C(s)$  is SNI, $C(s)$ is SNI.

Next consider the SISO  SNI controller
\begin{equation}
\label{resonant2}
C(s) = \sum_{i=1}^M\frac{-k_is(s+2\zeta_i\omega_{i})}{s^2 + 2\zeta_i\omega_{i} s +
  \omega_{i}^2},
\end{equation}
 where $\omega_{i} > 0$, $\zeta_i>0$, and $k_i>0$ for $i =
 1,2,\ldots,M$. Application of (\ref{resonant2}) is described
 in \cite{HM01,PMS02}.
By writing
\begin{equation}
\label{resonant_controller2}
\frac{-k_is(s+2\zeta_i\omega_{i})}{s^2 + 2\zeta_i\omega_{i} s +
  \omega_i^2} = -k_i+ \frac{k_i \omega_i^2}{s^2 + 2\zeta_i\omega_i s +
  \omega_i^2}
\end{equation}
for each $i$, it follows that the controller \eqref{resonant2}
is SNI. This result follows from the fact that  the first term on the right side of
(\ref{resonant_controller2}) has zero imaginary part, and the second
term on the right side of (\ref{resonant_controller2}) is
SNI as in the case of the positive-position feedback
controller (\ref{siso-ppf}). Using these facts, it follows from Lemma \ref{L1} that the controller (\ref{resonant2}) is SNI.

The SNI controllers (\ref{resonant1}) and
(\ref{resonant2}) can be  extended to the MIMO case
 to obtain the  MIMO
SNI controller
\eqn{\label{a}
C(s) = \sum_{i=1}^M\frac{-s^2}{s^2 + 2\zeta_i\omega_i s +
  \omega_i^2}\alpha_i\alpha_i^{\transpose}
}
and
\eqn{\label{b}
C(s) = \sum_{i=1}^M\frac{- s(s+2\zeta_i\omega_i) }{s^2 + 2\zeta_i\omega_i s +
  \omega_i^2}\beta_i\beta_i^{\transpose},
}
where  $\alpha_i$ and $\beta_i$ are $m \times 1$ vectors
\cite{MVB06}.
Control systems for flexible
structures with colocated force actuators and position sensors  using controllers
of the form \eqref{a}, \eqref{b} are {\em resonant
control systems} \cite{HM01,PMS02}.

\subsection{Integral Resonant Control}
Theorem \ref{IRC_SNI} shows that
MIMO transfer function matrices of the form
\[
C(s) = [sI+\Gamma \Phi]^{-1}\Gamma
\]
are SNI. Here,
$\Gamma$ is a positive-definite matrix and $\Phi$ is a
positive-definite matrix. The use of a controller of this
form when applied to a flexible structure with force actuators and
position sensors is referred to as {\em integral resonant control}, or
integral force control
\cite{PRE02,AFM07,BMP1a}.

To illustrate Theorem \ref{thm:main_result} and  integral resonant
control, consider a SISO
integral resonant
control system where the plant is a flexible structure with colocated force
actuation and position measurement. The plant is assumed to have the
transfer function
\eqn{\label{plantt}
P(s) = \sum_{k=1}^{10}\frac{1}{s^2+2s+10^{4}k^2}.
}

Now consider
 this system controlled with the SISO integral resonant controller
\begin{equation}\label{contrr}
C(s) = \frac{\Gamma}{s+\Gamma \Phi},
\end{equation}
where $\Gamma > 0$ and $\Phi > 0$. It follows from Theorem \ref{IRC_SNI} that (\ref{contrr}) is SNI. Using
Theorem \ref{thm:main_result}, it follows that the closed-loop system
is internally stable if the dc gain condition (\ref{DC_gain}) is satisfied. The dc value of the  plant
transfer function is $P(0) =  \sum_{k=1}^{10}\frac{1}{10^{4}k^2} =
1.5498\times 10^{-4}$, while the dc value of the controller transfer function is $C(0) =
1/\Phi$. By choosing $\Phi = 1.2 \times P(0) = 1.8597\times
10^{-4}$, the  condition $\Maeig[P(0)C(0)]  < 1$  is  satisfied. To choose the
parameter $\Gamma > 0$,  Figure \ref{F5.2b} shows the root locus of the closed-loop poles
for the  feedback control system with plant $P(s)$ given by \eqref{plantt}
and controller $C(s)$ given by \eqref{contrr} as the parameter $\Gamma > 0$ is varied. From
this root locus diagram, the parameter $\Gamma$ is chosen as $\Gamma =
9.6584\times 10^{5}$ to maximize the damping of the first resonant mode.

The damping of the resonant modes arising from the
integral resonant feedback controller \eqref{contrr} is illustrated in Figure
\ref{F5.2c}, which
shows the open-loop frequency response of the plant from the actuator
input to the sensor output. Also shown is the
closed-loop frequency response from the command input to the sensor output when
the integral resonant feedback controller $C(s) = \frac{9.6584\times
  10^{5}}{s+179.6379}$ is applied as in Figure
\ref{F5.2a}.

\subsection{State-Feedback Controller Synthesis}

An alternative approach to the direct use of Theorem \ref{thm:main_result} for
establishing the
closed-loop stability of a feedback control system is to
design the controller  to be robust against only a
specific uncertainty structure as shown in Figure \ref{F5.3}. In this case,
it  follows from Theorem \ref{thm:main_result}
that if the plant  uncertainty
 is known to be SNI, and if the
feedback controller is constructed so that the nominal closed-loop
system is NI and  the  dc  gain condition is satisfied,
then the resulting closed-loop
uncertain system is guaranteed to be robustly stable \cite{LP06}. We now present some further
results on this problem when full state-measurements are available using an LMI
approach. The assumption of full state-measurements means that there is a sensor available to measure each of the quantities that define a state variable in the state space model of the nominal plant shown in Figure  \ref{F5.3}.

Consider the  feedback control system  in Figure
\ref{F5.3} in the case that  full state feedback is available. In
  this case, Theorem
\ref{NIL} can be used to synthesize a state-feedback control law such
that the resulting closed-loop system is NI. Indeed, suppose  the  uncertain system shown in Figure
\ref{F5.3} is described by the state
equations
\begin{eqnarray}
\label{sys_ol}
&\dot x = Ax +B_1w+B_2u,& \\
\label{sys_ola}
&z = C_1x,&  \\
\label{sys_olb}
&w = \Delta(s)z,&
\end{eqnarray}
where  the uncertainty transfer function matrix  $\Delta(s)$ is assumed
to be  SNI with $|\Maeig(\Delta(0))|
\leq 1$ and $\Delta(\infty) \geq 0$.
Applying the state-feedback control law $u = Kx$ yields
the closed-loop uncertain system
\begin{eqnarray}
\label{sys_cl}
&\dot x = (A+B_2K)x +B_1w,& \\
\label{sys_cla}
&z = C_1x,&  \\
\label{sys_clb}
&w = \Delta(s)z.&
\end{eqnarray}
The corresponding nominal closed-loop transfer
function matrix is
\begin{equation}
\label{Gcl}
G_{\mbox{\small cl}}(s) = C_1(sI-A-B_2K)^{-1}B_1.
\end{equation}

\begin{theorem}
\label{SFNI_DCgain}
Consider the uncertain system (\ref{sys_ol}),  (\ref{sys_ola}),  (\ref{sys_olb}) and
suppose there exist  matrices $Y > 0$,
 $M$, and
 a scalar  $\varepsilon > 0$ such that
\begin{eqnarray}
& \left[\begin{array}{cc}
\label{lmi2}
AY+YA^{\transpose}+ B_2M+M^{\transpose}B_2^{\transpose} + \varepsilon I & B_1+AYC_1^{\transpose}+B_2MC_1^{\transpose} \\
B_1^{\transpose}+C_1YA^{\transpose}+C_1M^{\transpose}B_2^{\transpose} & 0
\end{array}\right] \leq   0,&\\
\label{lmi3}
 & C_1 Y C_1^T -I <  0,& \\
\label{lmi4}
&Y > 0.&
\end{eqnarray}
Here the parameter $\varepsilon > 0$ is chosen to be sufficiently small.
Then the state-feedback control law $u =  MY^{-1}x$ is robustly
stabilizing for the uncertain system (\ref{sys_ol}),  (\ref{sys_ola}),  (\ref{sys_olb}).
\end{theorem}

\begin{proof}
Suppose
the LMIs (\ref{lmi2}), (\ref{lmi4}) are  satisfied and
let
\[
K=MY^{-1}.
\]
Then, (\ref{lmi2}) implies
\begin{eqnarray}
\label{lmis1}
&(A+B_2K)Y+Y(A+B_2K)^{\transpose} = AY+YA^{\transpose}+ B_2M+M^{\transpose}B_2^{\transpose} + \varepsilon I \leq 0,&  \\
\label{lmis1a}
& B_1+(A+B_2K)YC_1^{\transpose} = B_1+AYC_1^{\transpose}+B_2MC_1^{\transpose} = 0.&
\end{eqnarray}
It follows from (\ref{lmis1}) that $A+B_2K$ has no eigenvalues on the imaginary axis. Furthermore, Theorem \ref{NIL} implies that
the closed-loop transfer function $G_{\mbox{\small cl}}(s)$ (\ref{Gcl}) is NI.

We now show that the  feedback system defined by (\ref{sys_cl}),
(\ref{sys_cla}), (\ref{sys_clb}), corresponding
to the state feedback control law $u = MY^{-1}x$, satisfies the assumptions of Theorem \ref{thm:main_result}.
Since $G_{\mbox{\small cl}}(s)$  is strictly proper, it follows that $G_{\mbox{\small cl}}(\infty) = 0$ and hence $\Delta(\infty)G_{\mbox{\small cl}}(\infty) =
0$. Also, it follows from (\ref{lmis1a}) that
\[
G_{\mbox{\small cl}}(0) = -C_1\left(A+B_2K\right)^{-1}B_1 = C_1 Y C_1^T.
\]
Therefore,  the LMI (\ref{lmi3}) implies $G_{\mbox{\small cl}}(0)< I$ and hence
\begin{equation}
\label{sigmaxGcl}
\sigma_{\mbox{\small max}}(G_{\mbox{\small cl}}(0)) < 1.
\end{equation}
However, since  $G_{\mbox{\small cl}}(s)$ is  negative imaginary, it follows from Lemma \ref{L5} that $G_{\mbox{\small cl}}(0) \geq G_{\mbox{\small cl}}(\infty) $. Moreover,  $G_{\mbox{\small cl}}(\infty) = 0$ and thus $G_{\mbox{\small cl}}(0) \geq 0$. Hence, $\Maeig(G_{\mbox{\small cl}}(0))
=\sigma_{\mbox{\small max}}(G_{\mbox{\small cl}}(0))$, and consequently $|\Maeig(G_{\mbox{\small cl}}(0))| < 1$. Also,
the
assumptions on
$\Delta(s)$ in (\ref{sys_ol}),  (\ref{sys_ola}),  (\ref{sys_olb}) imply that $\Delta(\infty) \geq 0$ and $|\Maeig(\Delta(0))| \leq
1$. From these conditions, it follows that
 $\Maeig(\Delta(0)G_{\mbox{\small cl}}(0)) < 1$.

Thus, we have
$\Delta(\infty)G_{\mbox{\small cl}}(\infty) = 0$, $\Delta(\infty) \geq 0$,
and
$
\Maeig(\Delta(0)G_{\mbox{\small cl}}(0)) < 1.
$
Therefore,  the assumptions of Theorem \ref{thm:main_result} are satisfied.
Now Theorem \ref{thm:main_result} implies that
the closed-loop system
(\ref{sys_ol}),  (\ref{sys_ola}),  (\ref{sys_olb}) with the state-feedback controller $u = MY^{-1}x$
is robustly stable.
\end{proof}

\subsection{An LMI State-Feedback Synthesis Example}

To   illustrate  Theorem \ref{SFNI_DCgain}, consider the system shown in Figure
\ref{F5.4}, which includes a flexible structure. The force
applied to the  flexible structure is denoted by
$x_2$, and the deflection of the structure at the same location is
denoted $y$. The transfer function from $x_2$ to $y$ is denoted
$G(s)$. The flexible structure has a colocated force
actuator and position sensor and the transfer function
$G(s)$ is  assumed to be SNI.
It is desired to construct a state
feedback controller for this system, which is robust against
unmodeled
flexible dynamics. Indeed, in order to apply the method of Theorem
\ref{SFNI_DCgain} to this example, the transfer function
$G(s)$ is replaced by a constant unity gain,  and the resulting
error is the SNI transfer function
$\Delta(s) = G(s) -1$. The transfer
function $\Delta(s)$  is  treated as
an uncertainty in the system as shown in Figure \ref{F5.5}.
A state-space realization of this uncertain system is
\begin{eqnarray*}
& \left[\begin{array}{c}
\dot x_1 \\ \dot x_2 \\ \dot x_3
\end{array}\right] =
 \left[\begin{array}{rrr}
-1 & 0 & 0\\
1 & -1 & 1\\
0 & 1 & -1
\end{array}\right]
 \left[\begin{array}{c}
 x_1 \\  x_2 \\  x_3
\end{array}\right]+
\left[\begin{array}{c}
 0 \\  0 \\  1
\end{array}\right]w+
\left[\begin{array}{r}
 -2 \\  1 \\  0
\end{array}\right]u,&\nonumber \\
&z = \left[\begin{array}{ccc}
0 & 1 & 0
\end{array}\right] \left[\begin{array}{c}
 x_1 \\  x_2 \\  x_3
\end{array}\right],& \nonumber \\
&w = \Delta(s)z.&
\end{eqnarray*}
Then, Theorem \ref{SFNI_DCgain} can be applied with
\begin{eqnarray*}
&A =
 \left[\begin{array}{rrr}
-1 & 0 & 0\\
1 & -1 & 1\\
0 & 1 & -1
\end{array}\right],~~
B_1=\left[\begin{array}{c}
 0 \\  0 \\  1
\end{array}\right], ~~
B_2=\left[\begin{array}{c}
  -2 \\  1 \\  0
\end{array}\right],& \\
&C_1= \left[\begin{array}{ccc}
0 & 1 & 0
\end{array}\right].&
\end{eqnarray*}
To apply  Theorem \ref{SFNI_DCgain}, we choose $\varepsilon = 10^{-6}$. Then the
LMIs (\ref{lmi2}) - (\ref{lmi4}) are solved using  LMI
software \cite{BCPS09} to find the
matrices $Y$ and  $M$ as
\begin{eqnarray*}
&Y =  \left[\begin{array}{ccc}
   3.9594 \times 10^9 & -2.0008 &  -3.9594\times 10^9 \\
  -2.0008 &   0.72850 &   1.7293\\
  -3.9594\times 10^9  & 1.7293 &   3.9594\times 10^9
\end{array}\right] > 0,& \\
&M = \left[\begin{array}{ccc}
-2.8122 &   1.0000 &  2.6260\end{array}\right].&
\end{eqnarray*}
Therefore, using Theorem \ref{SFNI_DCgain}, the required state
feedback gain matrix $K$ can be constructed as
\[
K = MY^{-1} = \left[\begin{array}{ccc}
0.22927 &   1.4581  & 0.22927
\end{array}\right].
\]
The Bode plot of the corresponding closed-loop transfer function from
$w$ to $z$, given by (\ref{Gcl})
is shown in Figure \ref{F5.6}.
From this Bode plot, it is seen that $G_{\mbox{\small cl}}(s)$ is  SNI since
\[
\angle G_{\mbox{\small cl}}(\jw) \in (-\pi,0)
\]
for all $\omega >0$.
Also, the Bode plot of Figure 12 shows that the magnitude of the dc value of $G_{\mbox{\small cl}}(s)$
is less than unity. Since the uncertainty transfer function
$\Delta(s)$ in this example is SNI,
it  follows from Theorem \ref{thm:main_result} that if $|\Delta(0)| \leq 1$, then the
closed-loop system is internally stable.

In the above example, the nominal system is obtained by
replacing the flexible structure transfer function $G(s)$ by a fixed
unity gain. This  gain can be regarded as an approximation
of the dc value of the flexible structure transfer function  $G(0)$. If  the
dc value of the flexible structure transfer function
 is known to be exactly unity, then it follows that $\Delta(0)
= G(0) - 1 =0$. In this case, the dc gain condition in Theorem
\ref{thm:main_result} is automatically satisfied, and there is
no need to require the  LMI condition (\ref{lmi3}) in constructing the
state-feedback controller. However, the current approach means that
the dc value of the flexible structure transfer function does not have
to be known exactly, and the control system is
robust against uncertainty in $G(0)$.

\section{Conclusion}
This article describes properties of a class of systems termed
NI systems using ideas from classical control
theory. Connections to positive-real and passive systems
are also given. It is also shown that the
class of NI systems  yields a
robust stability analysis result, which broadly speaking can be captured
by saying that if one system is negative imaginary and the other system
is strictly negative imaginary, then a necessary and sufficient
condition for internal stability of the positive-feedback
interconnection of  the two  systems is  that  the dc loop gain
is less than unity. This result provides a  natural
framework for the analysis of robust stability of lightly damped flexible structures
with unmodeled dynamics. This result also captures, in a systematic
framework, graphical design methods adopted in the 1980s by practical
engineers related to positive-position feedback and similar
techniques. This article also provides a full state-feedback controller
synthesis technique that achieves a
NI closed-loop system. The use of this theory is similar to the use of passivity
theory, and hence extends and complements
existing
passivity results \cite{DV75,KHA01}.
\section{Acknowledgments}
The authors wish to acknowledge many discussions on the
topic of this article with Professor Reza Moheimani, Dr. Junlin
Xiong, Dr. Sourav Patra and Zhuoyue Song. They also wish to acknowledge financial support from the
Australian Research Council, the Engineering and Physical Sciences Research Council and the Royal Society.

\newpage
\clearpage

\eject
\newpage
\mbox{}

\begin{table}
\begin{center}
\begin{tabular}{|l|l|}
\hline
$A^*$ & Complex conjugate transpose of the complex matrix
$A$ \\
$A^{\transpose}$ &  Transpose of the  matrix
$A$ \\
$A>0$ & The matrix $A$ is positive definite \\
$A\geq 0$ &  The matrix $A$ is positive semidefinite\\
$\Re[s]$ & Real part of the complex number $s$ \\
$\Im[s]$ & Imaginary part of the complex number $s$ \\
$\Im_{\text{H}}[A]$ & Hermitian-Imaginary part $\jmath[A-A^*]$ \\
$\Maeig(A)$ &  Maximum eigenvalue of the
  matrix $A$  whose eigenvalues are all  real.\\
$\sigma_{\mbox{\small max}}(A)$ & Maximum singular value of the matrix $A$. \\
CRHP & closed right half of the complex plane \\
ORHP & open right half of the complex plane \\
CLHP & closed left half of the complex plane \\
OLHP & open left half of the complex plane \\
\hline
\end{tabular}
\end{center}
\caption{Notation.}
\label{notation}
\end{table}

\eject
\newpage
\mbox{}

\begin{table}
\begin{center}
\begin{tabular}{|l|l|}
\hline
SISO & single-input, single-output\\
MIMO &  multi-input, multi-output  \\
NI & negative-imaginary \\
SNI & strictly  negative-imaginary \\
LMI & linear matrix inequality \\
RLC & resistor, inductor, capacitor \\
\hline
\end{tabular}
\end{center}
\caption{List of Acronyms.}
\label{acronyms}
\end{table}

\newpage
\begin{figure}[H]
\psfrag{M(jw)}{$P(\jw)$}
\centering
\includegraphics[scale=1.25]{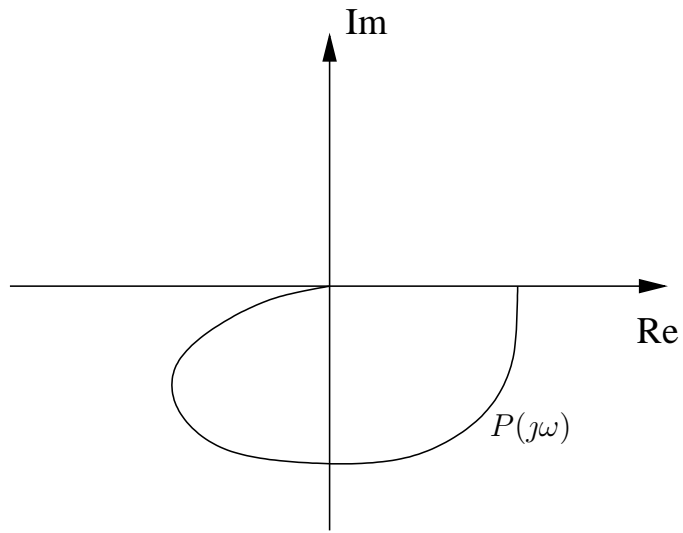}
\caption{Positive-frequency Nyquist plot of a negative-imaginary system. A single-input, single-output
  negative-imaginary  transfer function has no poles in CRHP and has a frequency response
  with negative imaginary part for all frequencies. Consequently,  the
  Nyquist plot for $\omega>0$
  is contained  in the lower half  of the complex plane.}
\label{one-NIFR-system}
\end{figure}

\newpage
\begin{figure}[H]
\centering
\psfrag{M}{$M(s)$}
\psfrag{N}{$N(s)$}
\psfrag{w1}{$w_1$}
\psfrag{u1}{$u_1$}
\psfrag{y1}{$y_1$}
\psfrag{w2}{$w_2$}
\psfrag{u2}{$u_2$}
\psfrag{y2}{$y_2$}
\includegraphics[scale=1.25]{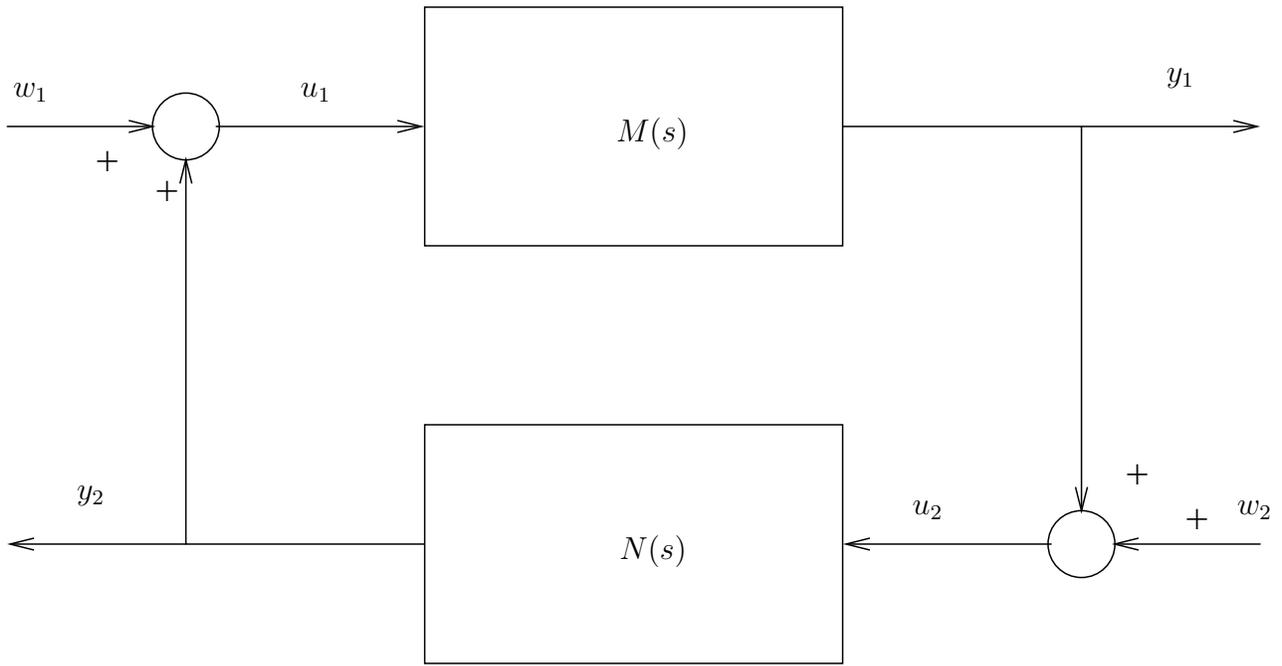}
\caption{A positive feedback interconnection. The transfer functions
  $M(s)$ and $N(s)$ are interconnected by
  positive feedback. The stability of this feedback interconnection
  can be guaranteed by using either the small-gain theorem or Theorem \ref{thm:main_result}. The relevant stability result
   depends on the properties of  $M(s)$
  and $N(s)$.}
\label{feedback-interconnection}
\end{figure}
\newpage
\begin{figure}[H]
\centering
\setlength{\unitlength}{0.00083333in}
\begingroup\makeatletter\ifx\SetFigFontNFSS\undefined%
\gdef\SetFigFontNFSS#1#2#3#4#5{%
  \reset@font\fontsize{#1}{#2pt}%
  \fontfamily{#3}\fontseries{#4}\fontshape{#5}%
  \selectfont}%
\fi\endgroup%
{\renewcommand{\dashlinestretch}{30}
\begin{picture}(6428,4821)(0,-10)
\path(2265,4737)(4065,4737)(4065,3312)
	(2265,3312)(2265,4737)
\path(2265,1437)(4065,1437)(4065,12)
	(2265,12)(2265,1437)
\path(5865,4512)(4065,4512)
\blacken\path(4185.000,4542.000)(4065.000,4512.000)(4185.000,4482.000)(4149.000,4512.000)(4185.000,4542.000)
\path(5865,312)(4065,312)
\blacken\path(4185.000,342.000)(4065.000,312.000)(4185.000,282.000)(4149.000,312.000)(4185.000,342.000)
\path(2265,312)(390,312)
\blacken\path(510.000,342.000)(390.000,312.000)(510.000,282.000)(474.000,312.000)(510.000,342.000)
\path(2265,4512)(465,4512)
\blacken\path(585.000,4542.000)(465.000,4512.000)(585.000,4482.000)(549.000,4512.000)(585.000,4542.000)
\path(2265,3612)(1665,3612)(1665,2712)
	(4665,1812)(4665,1137)(4065,1137)
\blacken\path(4185.000,1167.000)(4065.000,1137.000)(4185.000,1107.000)(4149.000,1137.000)(4185.000,1167.000)
\path(2265,1137)(1665,1137)(1665,1812)
	(4665,2712)(4665,3612)(4065,3612)
\blacken\path(4185.000,3642.000)(4065.000,3612.000)(4185.000,3582.000)(4149.000,3612.000)(4185.000,3642.000)
\put(2940,3912){\makebox(0,0)[lb]{\smash{{\SetFigFontNFSS{12}{14.4}{\rmdefault}{\mddefault}{\updefault}M}}}}
\put(3090,537){\makebox(0,0)[lb]{\smash{{\SetFigFontNFSS{14}{16.8}{\rmdefault}{\mddefault}{\updefault}N}}}}
\put(6090,4587){\makebox(0,0)[lb]{\smash{{\SetFigFontNFSS{14}{16.8}{\rmdefault}{\mddefault}{\updefault}w1}}}}
\put(15,312){\makebox(0,0)[lb]{\smash{{\SetFigFontNFSS{14}{16.8}{\rmdefault}{\mddefault}{\updefault}y2}}}}
\put(1290,3087){\makebox(0,0)[lb]{\smash{{\SetFigFontNFSS{14}{16.8}{\rmdefault}{\mddefault}{\updefault}u2}}}}
\put(4890,3087){\makebox(0,0)[lb]{\smash{{\SetFigFontNFSS{14}{16.8}{\rmdefault}{\mddefault}{\updefault}u1}}}}
\put(6090,312){\makebox(0,0)[lb]{\smash{{\SetFigFontNFSS{14}{16.8}{\rmdefault}{\mddefault}{\updefault}w2}}}}
\put(15,4587){\makebox(0,0)[lb]{\smash{{\SetFigFontNFSS{14}{16.8}{\rmdefault}{\mddefault}{\updefault}y1}}}}
\end{picture}
}
\caption{A Redheffer star product feedback interconnection. If this
  feedback interconnection is internally stable and the
  transfer function matrices $M(s)$ and $N(s)$ are negative imaginary,
  then $T(s)$, the closed-loop transfer function matrix  from
 $\left[w_1^T ~~ w_2^T\right]^T$ to
 $\left[y_1^T ~~y_2^T\right]^T$, is
 negative imaginary.
Furthermore, if in addition, either
 $M(s)$ or $N(s)$ is  strictly negative imaginary, then $T(s)$ is
 strictly negative imaginary.
}
\label{star-product}
\end{figure}
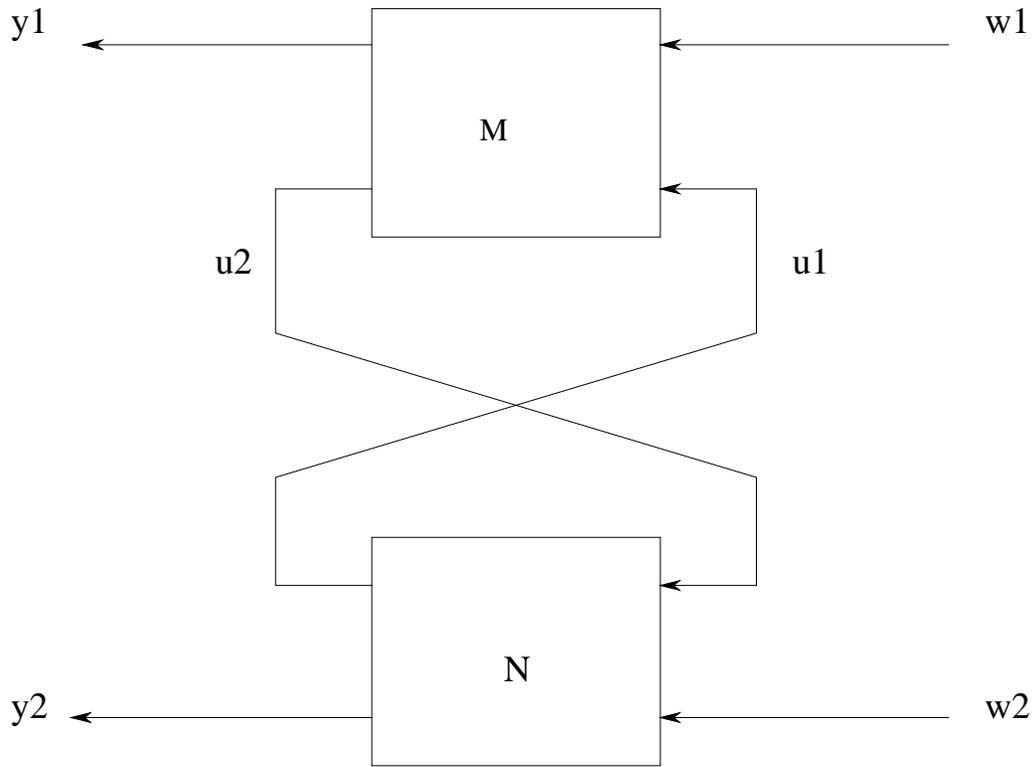

\newpage
\begin{figure}[H]
\centering
\psfrag{Nyquist Plot}{}
\psfrag{real}{Real (rad/s)}
\psfrag{imaginary}{Imaginary (rad/s)}
\includegraphics[width=16cm]{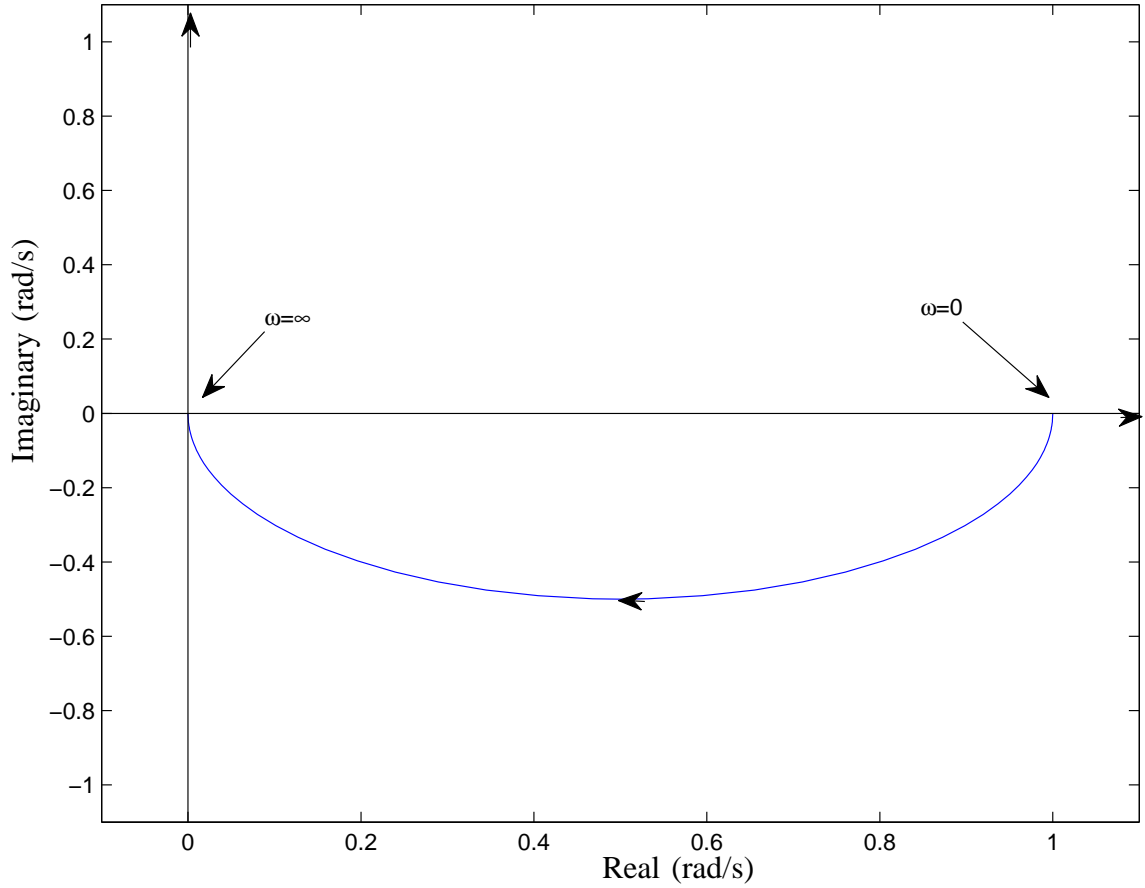}
\caption{Positive-frequency Nyquist plot of the transfer function $M(s) =
  \frac{1}{s+1}$.   The imaginary part of $M(\jw)$ is
   negative  for all $\omega > 0$, and thus $M(s)$ is  strictly
  negative imaginary.}
\label{F4.1}
\end{figure}
\newpage
\begin{figure}[H]
\centering
\psfrag{Nyquist Plot}{}
\psfrag{w=0}{$\scriptstyle \omega = 0$}
\psfrag{w=1}{$\scriptstyle\omega = 1$}
\psfrag{w=infty}{$\scriptstyle\omega = \infty$}
\psfrag{real}{Real (rad/s)}
\psfrag{imaginary}{Imaginary (rad/s)}
\includegraphics[width=16cm]{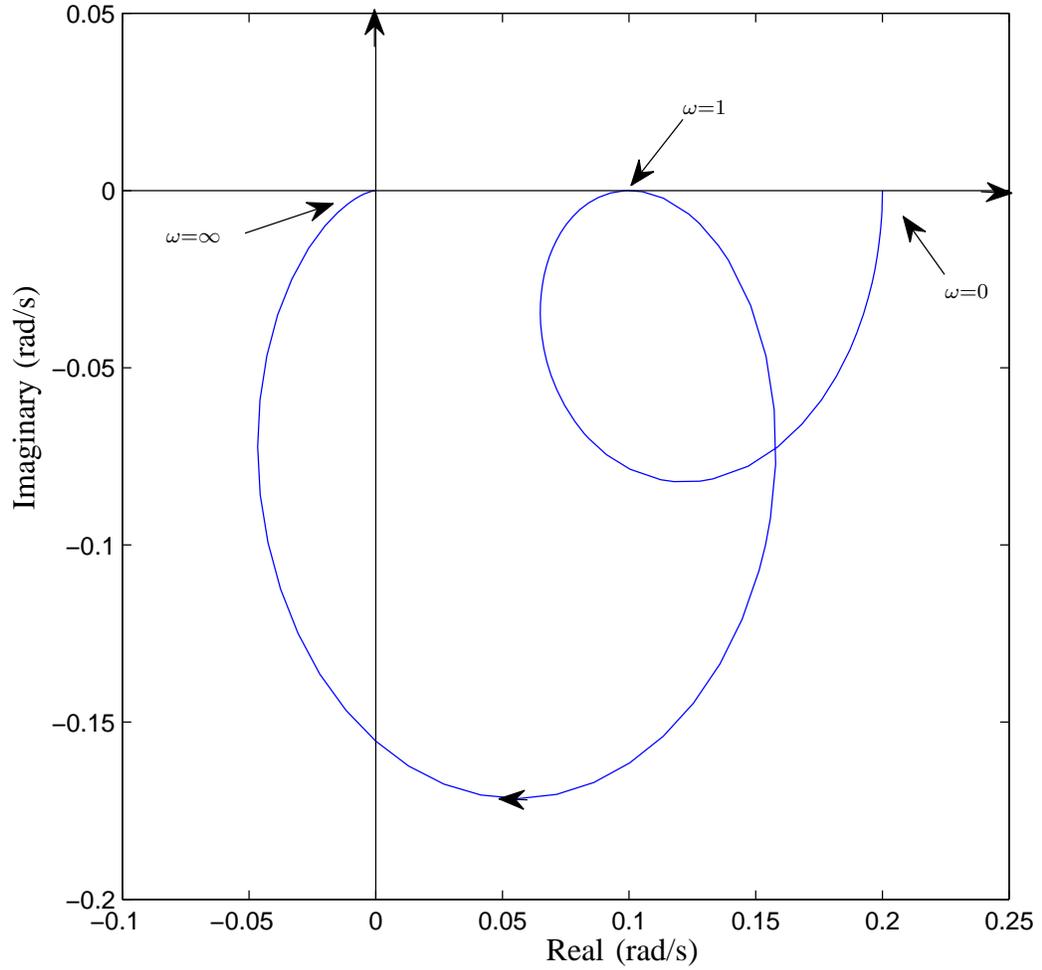}
\caption{Positive-frequency Nyquist plot of the transfer function $M(s) =
  \frac{2s^2+s+1}{(s^2+2s+5)(s+1)(2s+1)}$. This Nyquist plot
  shows  that the imaginary part of $M(\jw)$ is
   negative for all $\omega \geq  0$ except $\omega
  =0$ and $\omega
  =1$, where the imaginary part of $M(\jw)$ is zero. Thus $M(s)$ is negative
  imaginary, but not strictly negative imaginary.}
\label{F4.2}
\end{figure}
\newpage

\begin{figure}[H]
\centering
\psfrag{Nyquist Plot}{}
\psfrag{w=0}{$\omega = 0$}
\psfrag{w=infty}{$\omega = \infty$}
\psfrag{real}{Real (rad/s)}
\psfrag{imaginary}{Imaginary (rad/s)}
\includegraphics[width=16cm]{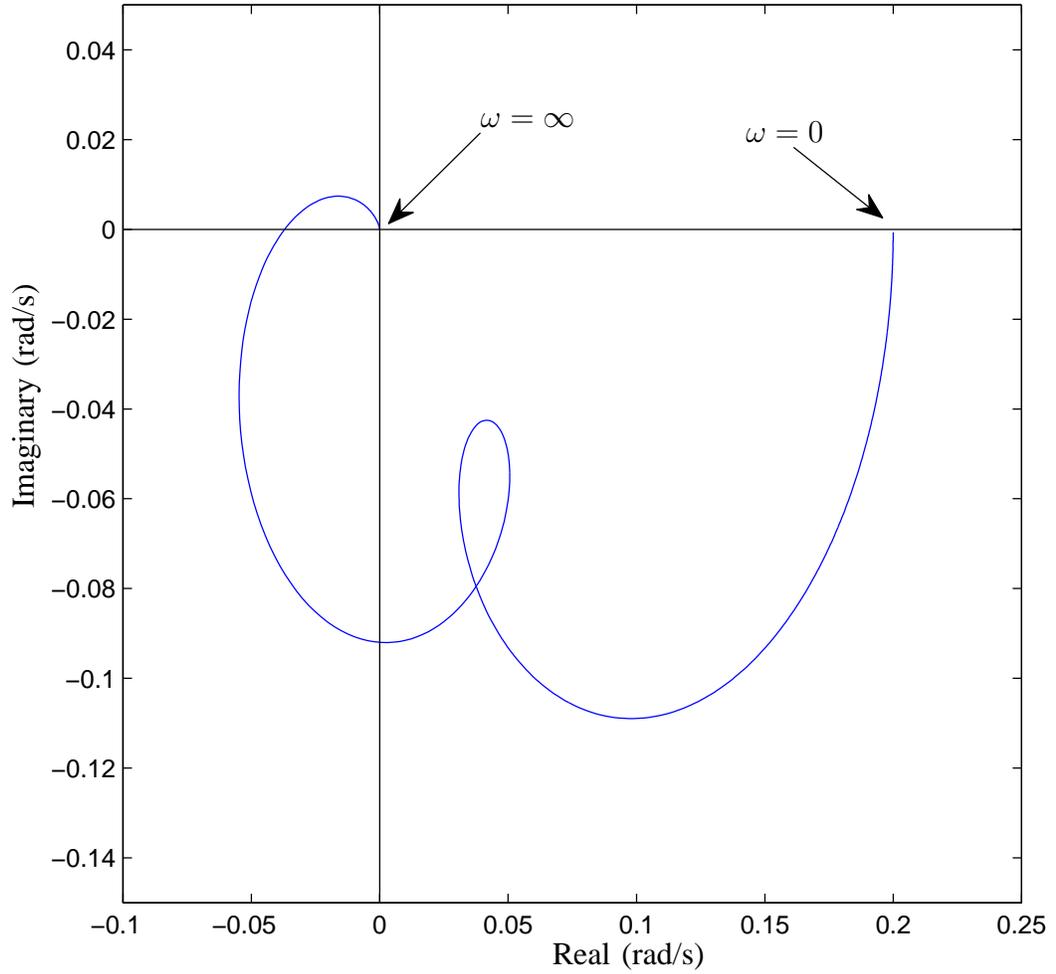}
\caption{Positive-frequency Nyquist plot of the loop transfer function $L(s) =
  N(s)M(s)$ corresponding to the positive-feedback interconnection
  of $M(s)= \frac{1}{s+1} $ and $N(s)=
\frac{2s^2+s+1}{(s^2+2s+5)(s+1)(2s+1)}$. Here  $M(s)$ is
strictly negative imaginary, and $N(s)$ is negative imaginary. Since $M(s)$ and
$N(s)$  both have no poles in CRHP  and the Nyquist plot does not encircle the
critical point $s=1 + \jmath 0$, it follows from the Nyquist stability
criterion that the positive feedback interconnection of $M(s)$ and $N(s)$ is internally stable.
 }
\label{F5.2}
\end{figure}
\newpage
\begin{figure}[H]
\centering
\psfrag{P(s)}{$P(s)$}
\psfrag{C(s)}{$C(s)$}
\psfrag{Stable}{}
\includegraphics[width=14cm]{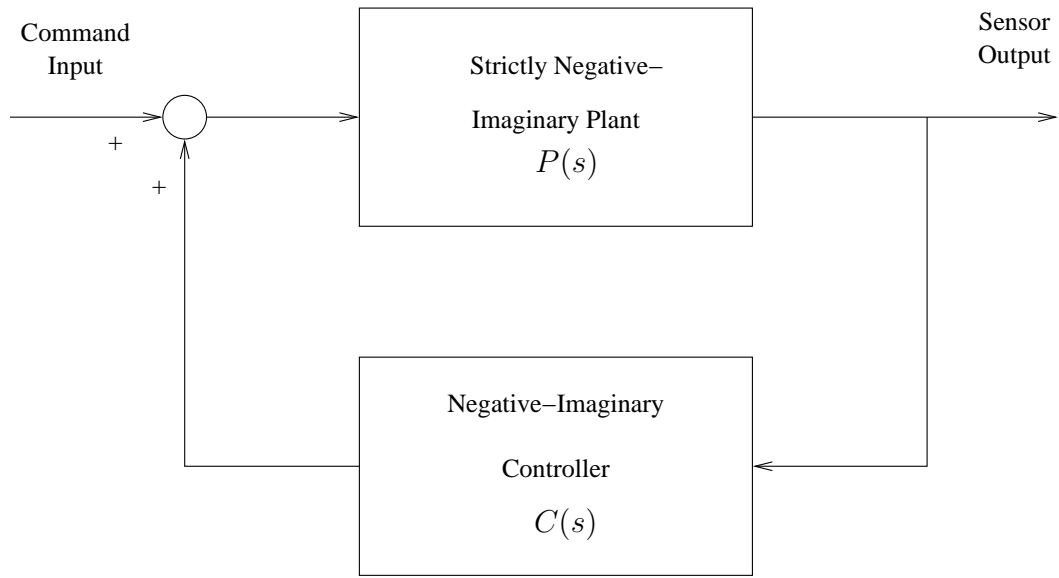}
\caption{Negative-imaginary  feedback control system. If the plant
  transfer function matrix $P(s)$ is strictly negative
  imaginary, and the  controller transfer
  function matrix $C(s)$ is negative imaginary, then the closed-loop
  system is internally stable if and only if the dc gain condition
  $\Maeig[P(0)C(0)]  < 1$ is satisfied.
 }
\label{F5.2a}
\end{figure}
\newpage
\begin{figure}[H]
\centering
\psfrag{Root Locus}{}
\psfrag{Real Axis}{Real  (rad/s)}
\psfrag{Imaginary Axis}{Imaginary  (rad/s)}
\includegraphics[width=16cm]{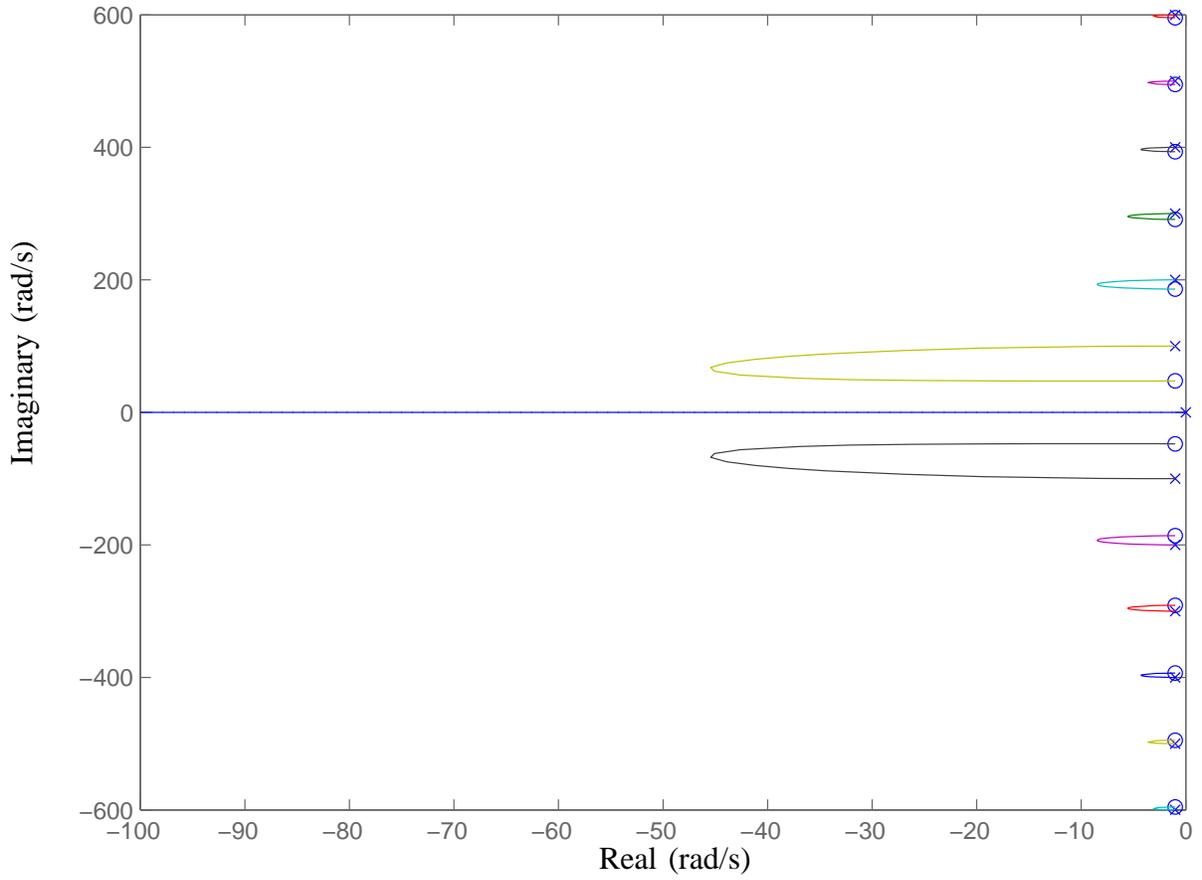}
\caption{Root locus of the closed-loop poles of a control system
  consisting of a flexible structure
  plant and an integral resonant
  controller. Here, the plant transfer function is $P(s) =
  \sum_{k=1}^{10}\frac{1}{s^2+2s+10^{4}k^2}$, and the
  integral resonant
  controller transfer function is $C(s) = \frac{\Gamma}{s+\Gamma
    \Phi}$. In this control system, both the plant and the controller
  are strictly negative imaginary. The root locus is obtained by varying the
  parameter $\Gamma > 0$ with
  $\Phi= 1.8597\times
10^{-4} \texttt{ m/N}$.
 }
\label{F5.2b}
\end{figure}
\newpage

\begin{figure}[H]
\centering
\psfrag{Magnitude dB}{Magnitude (dB)}
\psfrag{Frequency rad/s}{Frequency (rad/s)}
\includegraphics[width=16cm]{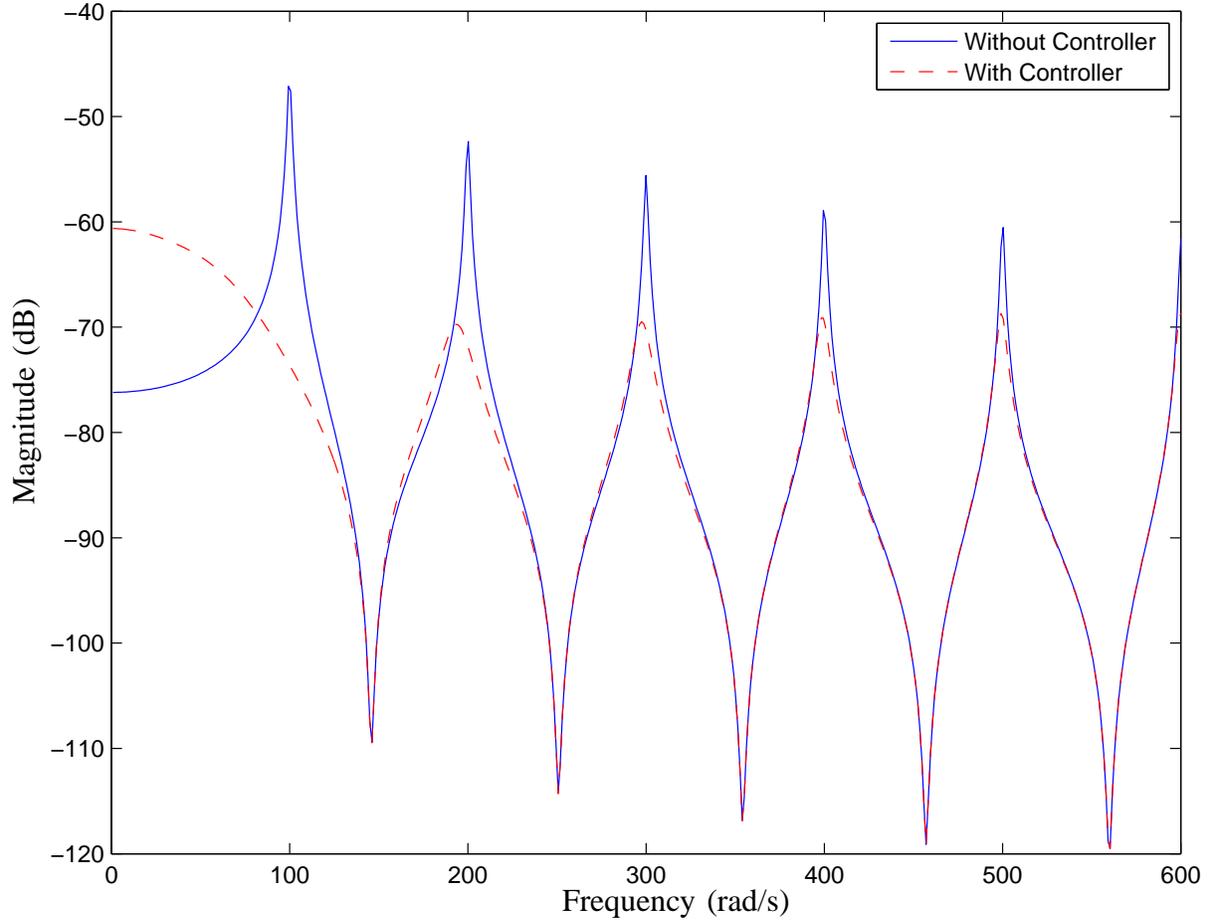}
\caption{
Open- and closed-loop frequency responses for a lightly damped flexible
structure with an integral resonant feedback
  controller. Here the plant transfer function is $P(s) =
  \sum_{k=1}^{10}\frac{1}{s^2+2s+10^{4}k^2}$, and the frequency
  response is taken from the command input to the sensor output. The
  closed-loop  uses the integral resonant feedback
  controller $C(s) = \frac{\Gamma}{s+\Gamma \Phi}$. The  parameter
  $\Gamma = 9.6584 \times 10^5 \texttt{ rad-N/(s-m)}$ is chosen to provide adequate damping
  of the low-frequency resonant modes, and the parameter $\Phi$ is
  fixed at $\Phi= 1.8597\times
10^{-4} \texttt{ m/N}$.
}
\label{F5.2c}
\end{figure}

\newpage
\begin{figure}[H]
\centering
\psfrag{D(s)}{$\Delta(s)$}
\psfrag{G(s)}{$G_{\mbox{\small cl}}(s)$}
\includegraphics[width=16cm]{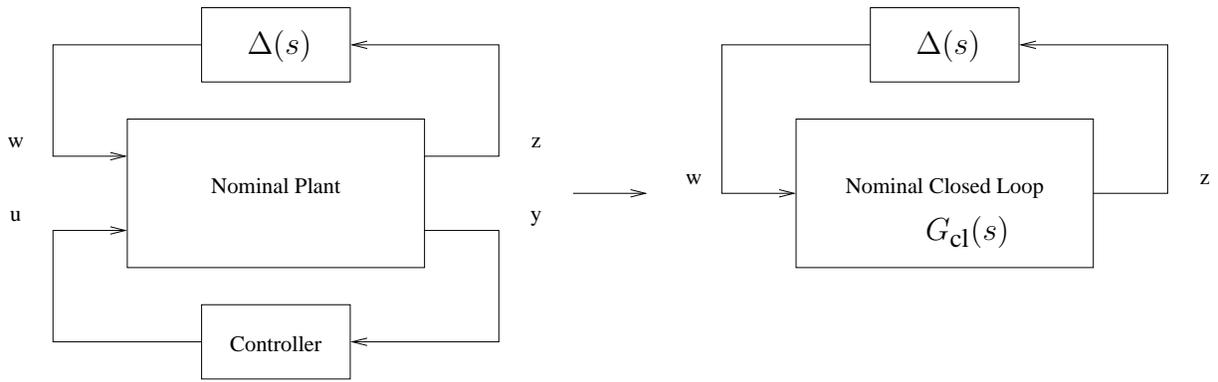}
\caption{A feedback control system.  The plant uncertainty
  $\Delta(s)$ is strictly negative imaginary, and satisfies the dc
  gain condition $\sigma_{\mbox{\small max}}(\Delta(0)) <
  \mu$ and $\Delta(\infty)\geq 0$. If the controller
  is chosen so that the nominal closed-loop transfer function matrix $G_{\mbox{\small cl}}(s)$ is strictly proper,
   negative
  imaginary, and satisfies the  dc gain condition $\sigma_{\mbox{\small max}}(G_{\mbox{\small cl}}(0))
  \leq   \frac{1}{\mu}$, then the closed-loop system is  robustly stable for all strictly negative imaginary
  uncertainty $\Delta(s)$ satisfying $\sigma_{\mbox{\small max}}(\Delta(0)) <
  \mu$ and $\Delta(\infty)\geq 0$.
 }
\label{F5.3}
\end{figure}
\newpage
\begin{figure}[H]
\centering
\psfrag{s+1}{$\scriptstyle s+1$}
\psfrag{y}{$y$}
\psfrag{F(s)}{$G(s)$}
\psfrag{u}{$u$}
\psfrag{x1}{$x_1$}
\psfrag{x2}{$x_2$}
\psfrag{x3}{$x_3$}
\includegraphics[width=16cm]{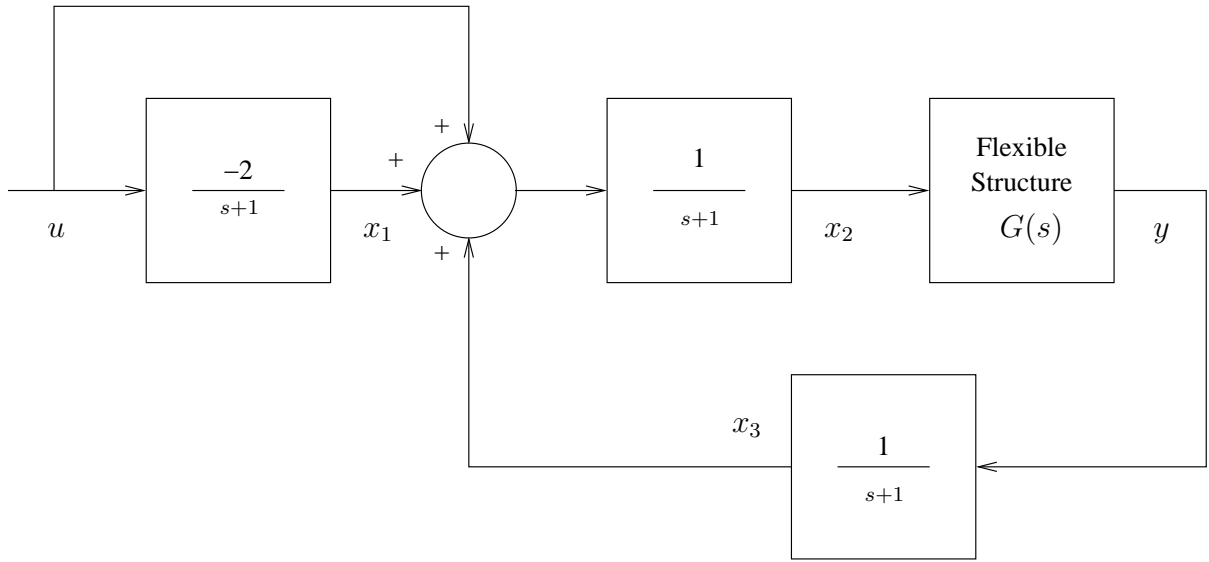}
\caption{Control of a flexible structure system using a
  state-feedback linear matrix inequality (LMI) approach to robust
  controller design. This system
  includes the
  unknown flexible structure transfer function $G(s)$. In this system,
  the force applied to the structure is labeled $x_2$, and the
  deflection of the structure at the same location is labeled $y$. A
  state-feedback controller is to be designed for this system by
  replacing
  the flexible structure transfer function $G(s)$ by a unity gain, and
  treating the
  resulting error $\Delta(s) = G(s) - 1$ as a strictly
  negative-imaginary uncertainty. The  state-feedback
  controller
gain matrix  can be obtained by solving an LMI feasibility problem.
 }
\label{F5.4}
\end{figure}
\newpage
\begin{figure}[H]
\centering
\psfrag{z}{$z$}
\psfrag{w}{$w$}
\psfrag{u}{$u$}
\psfrag{y}{$y$}
\psfrag{s+1}{$\scriptstyle s+1$}
\psfrag{x1}{$x_1$}
\psfrag{x2}{$x_2$}
\psfrag{x3}{$x_3$}
\psfrag{D}{$\Delta(s)$}
\includegraphics[width=16cm]{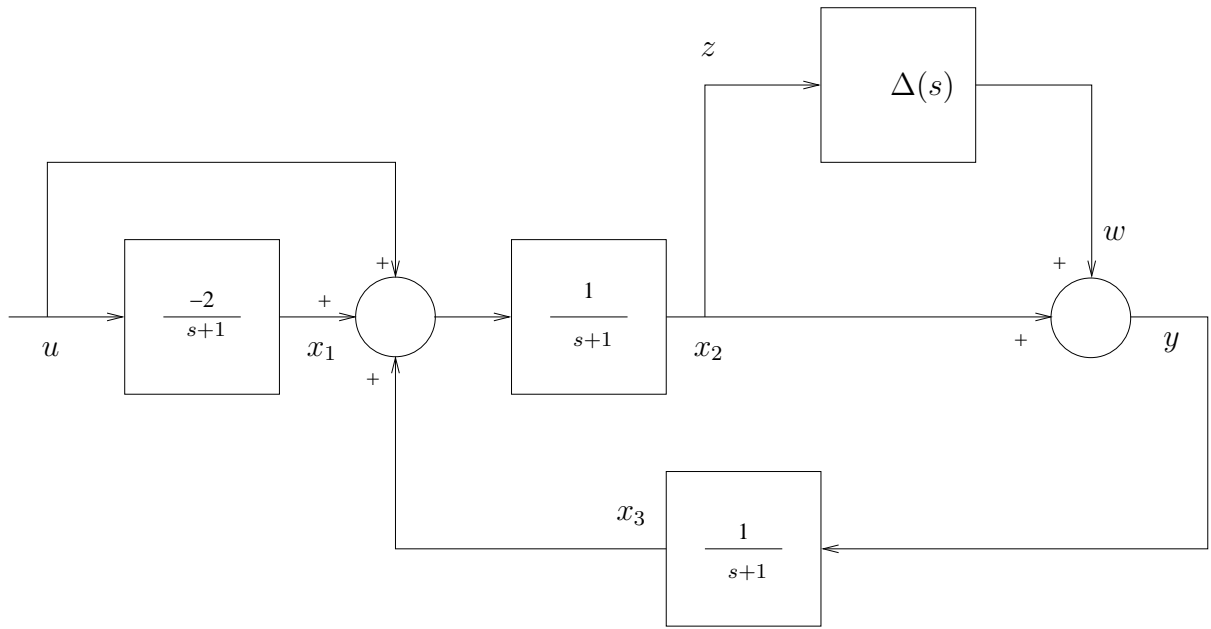}
\caption{Control of an uncertain system  using full-state-feedback control.  The uncertain system is constructed from the
  system shown in Figure \ref{F5.4} by
  replacing the flexible structure transfer
  function $G(s)$ by $1+\Delta(s)$, where $\Delta(s)$ is an uncertain but strictly negative-imaginary  transfer function. The signal $z$ is
treated as an uncertainty output and the signal $w$ is treated as an
uncertainty input.
 }
\label{F5.5}
\end{figure}
\newpage
\begin{figure}[H]
\centering
\psfrag{Bode Diagram}{}
\includegraphics[width=16cm]{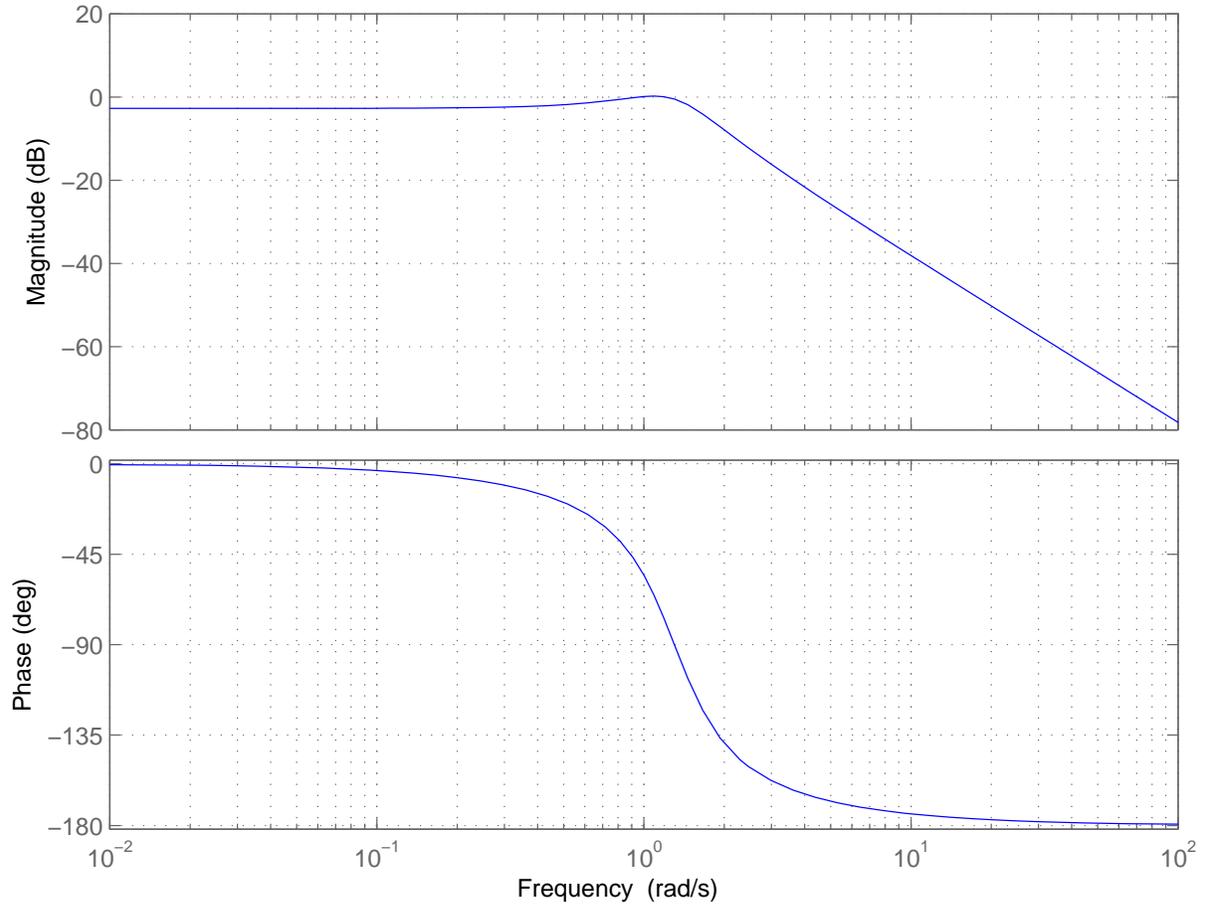}
\caption{Bode plot of the closed-loop transfer function $G_{\mbox{\small cl}}(s)$ from  the
  uncertainty input $w$ to the uncertainty output
  $z$. This closed-loop system is obtained from the  system shown in Figure \ref{F5.5} using a
  full-state-feedback controller obtained from Theorem \ref{SFNI_DCgain}. The fact that  $\angle G_{\mbox{\small cl}}(\jw) \in (-\pi,0)$ for all
   $\omega > 0$  implies that this transfer
  function is
  strictly negative imaginary. Also, since $G_{\mbox{\small cl}}(s)$ has no
  poles in CRHP and   $|G_{\mbox{\small cl}}(0)| < 1$, it follows that the closed-loop uncertain system is
  internally stable for all uncertainties $\Delta(s)$ that are
  strictly negative imaginary and satisfy $|\Delta(0)| < 1$. The fact that $|G_{\mbox{\small cl}}(0)| < 1$ c
  an be seen in the magnitude Bode plot.
 }
\label{F5.6}
\end{figure}

\newpage
\section{Sidebar 1\\ What Is Positive-real and Passivity Theory?}
\renewcommand{\thetheorem}{S\arabic{theorem}}
\setcounter{theorem}{0}
\renewcommand{\thedefinition}{S\arabic{definition}}
\setcounter{definition}{0}
\renewcommand{\thefigure}{S\arabic{figure}}
\setcounter{figure}{0}

A SISO positive-real transfer function    has a positive real part at all frequencies; a typical frequency response is depicted in
Figure~\ref{one-positive-real-system}.
The passivity theorem, which underpins
much of the robust and adaptive control literature~\cite{SAN65}, concerns the internal stability of
the negative-feedback interconnection, as shown in
Figure~\ref{negative_feedback-interconnection}, of two positive-real
transfer function matrices.
\begin{definition} (\cite{ZDG96})
\label{D2}
The feedback interconnection of two systems with transfer function
matrices $M(s)$ and $N(s)$ as shown in Figure
\ref{negative_feedback-interconnection} is  {\em internally
  stable} if the interconnection does not contain an algebraic loop
and the transfer
function matrix from exogenous signals to internal signals has no
poles in  CRHP.
\end{definition}

The following result is the {\em passivity theorem} \cite{DV75},\cite[Section 6.5]{KHA01}.

\begin{theorem}
The negative-feedback interconnection  of the positive-real
transfer function matrix $M(s)$ and the strictly positive-real transfer function matrix $N(s)$
is internally stable.
\end{theorem}

The SISO positive-real transfer function $M(s)$ satisfies
$\angle M(\jw)\in[-\pi/2,\pi/2]$ for
all  $\omega \geq 0$. Also, the SISO strictly positive-real
transfer function  $N(s)$ satisfies
$\angle N(\jw)\in(-\pi/2,\pi/2)$ for all  $\omega \geq
0$. From $\angle M(\jw)\in[-\pi/2,\pi/2]$  and $\angle N(\jw)\in(-\pi/2,\pi/2)$ for all  $\omega \geq
0$, it follows
that $\angle M(\jw)N(\jw)=\angle M(\jw)+\angle
N(\jw)\in(-\pi,\pi)$ for all  $\omega \geq 0$, and hence
the Nyquist plot of $M(\jw) N(\jw)$ cannot intersect
the negative real axis. Consequently, the  Nyquist plot of $M(s)N(s)$ cannot encircle the  Nyquist point $s=-1+\jmath 0$, and
internal stability of the negative-feedback  interconnection of $M(s)$
and $N(s)$  follows
from the  Nyquist stability criterion as depicted
in Figure~\ref{passivity}.

 The above concepts relating to  positive-real systems and the
 passivity theorem  generalize
 to MIMO linear time-invariant  systems and also to a
nonlinear and time-varying setting~\cite{DV75}.

\newpage
\begin{figure}[H]
\psfrag{M}{$M(\jw)$}
\centering
\includegraphics[scale=1.25]{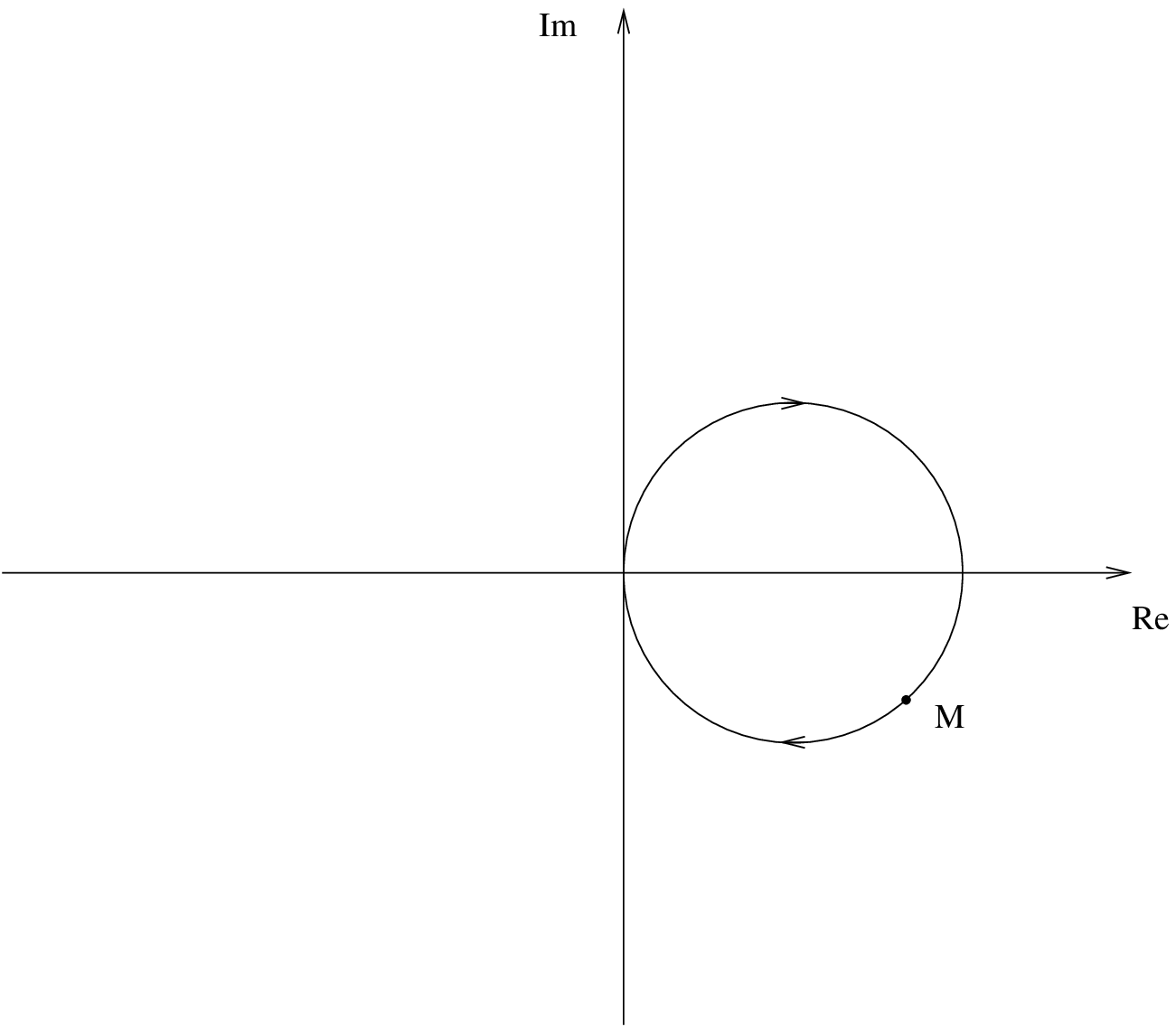}
\caption{The Nyquist plot of the positive-real transfer function $M(s)=\frac{1}{s+1}$. This plot
  illustrates the fact that, for a  single-input, single-output  positive-real
  transfer function,  the real part of its frequency response is positive for all frequencies.
  Consequently,  the Nyquist plot is contained in CRHP.}
\label{one-positive-real-system}
\end{figure}

\newpage
\begin{figure}[H]
\centering
\psfrag{M}{$M(s)$}
\psfrag{N}{$N(s)$}
\psfrag{w1}{$w_1$}
\psfrag{u1}{$u_1$}
\psfrag{y1}{$y_1$}
\psfrag{w2}{$w_2$}
\psfrag{u2}{$u_2$}
\psfrag{y2}{$y_2$}
\includegraphics[scale=1.25]{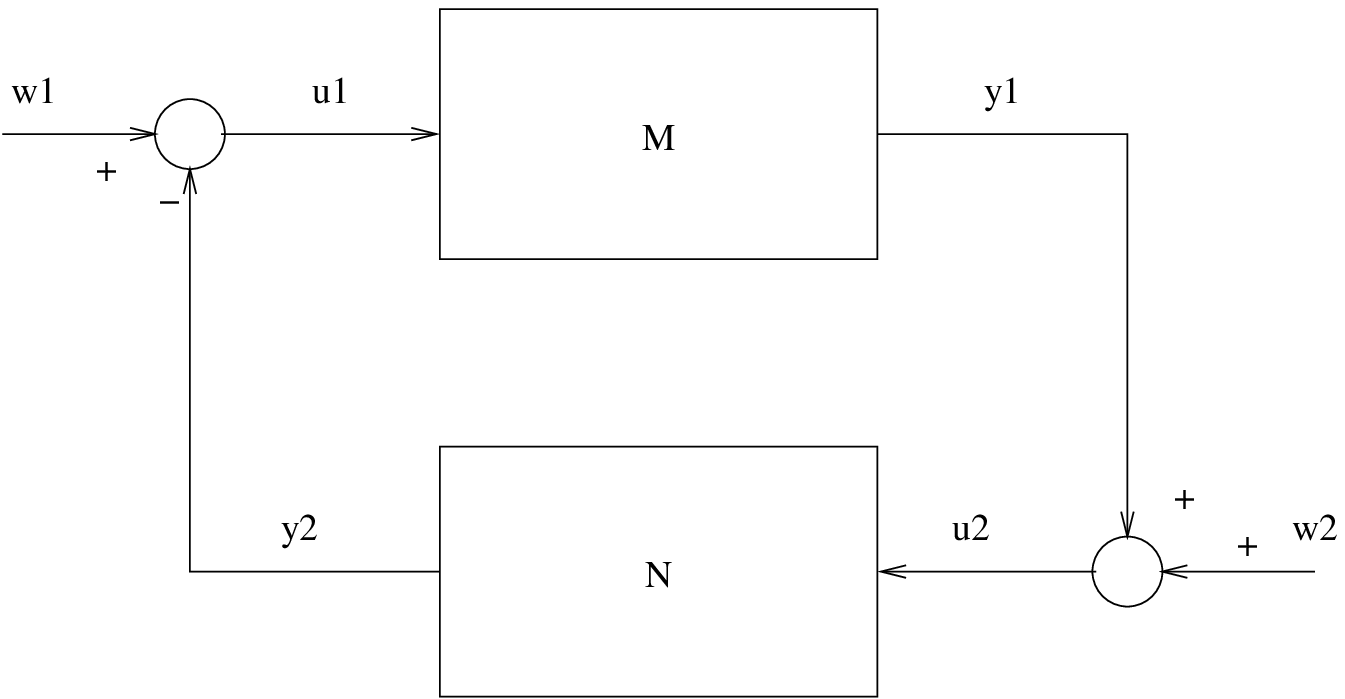}
\caption{A negative-feedback interconnection. This figure shows the
  negative-feedback interconnection of the transfer functions
  $M(s)$ and $N(s)$. The stability of this feedback
interconnection can be guaranteed using the passivity theorem if
$M(s)$ and $N(s)$ are positive real and either $M(s)$ or $N(s)$ is strictly
positive real.}
\label{negative_feedback-interconnection}
\end{figure}

\newpage
\begin{figure}[H]
\centering
\mbox{\psfrag{M}{$\scriptscriptstyle M(\jw)$}\psfrag{N}{$\scriptscriptstyle N(\jw)$}\psfrag{LG}{$\scriptscriptstyle M(\jw)N(\jw)$}\includegraphics[scale=.4]{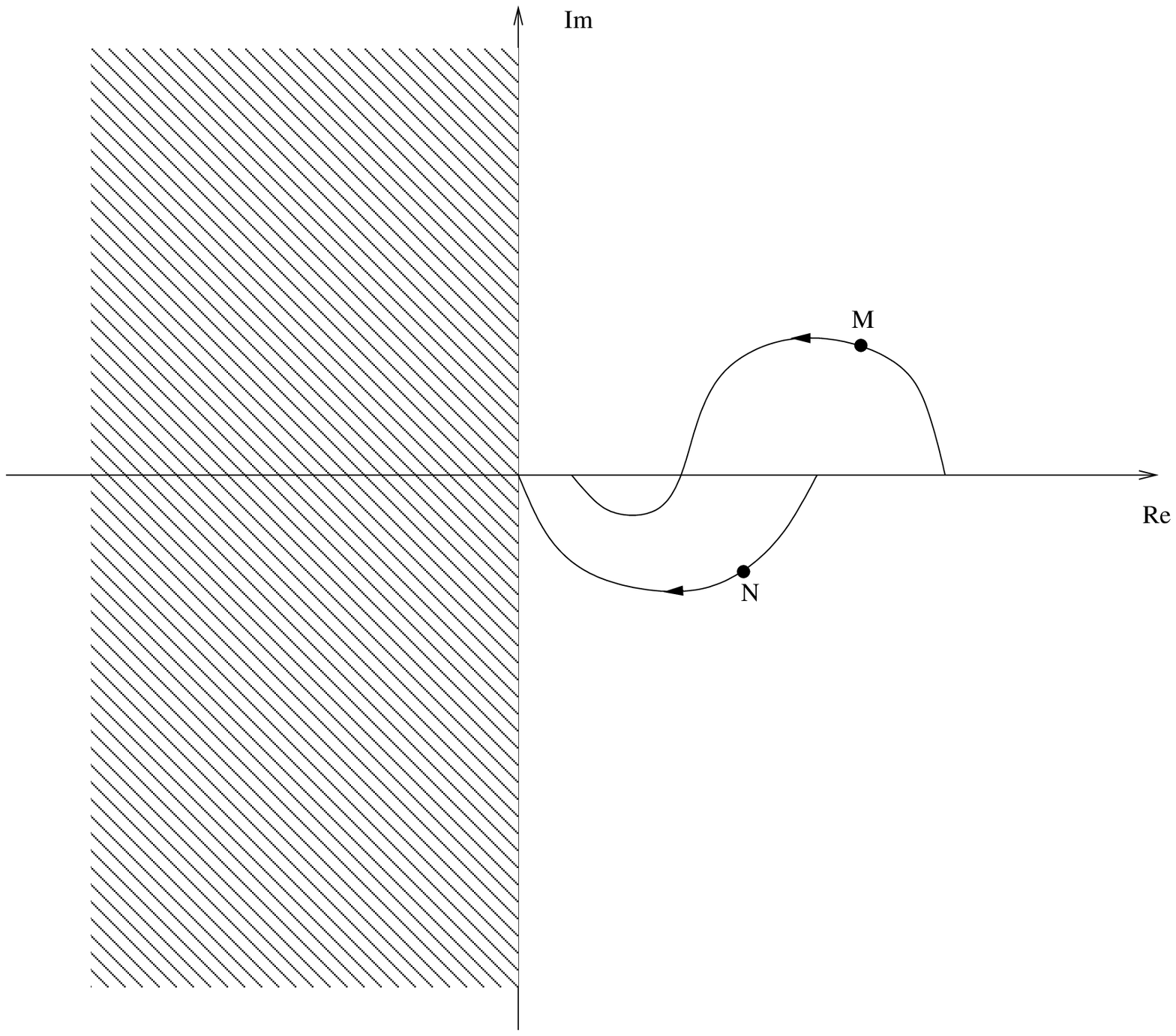}\quad
\includegraphics[scale=.4]{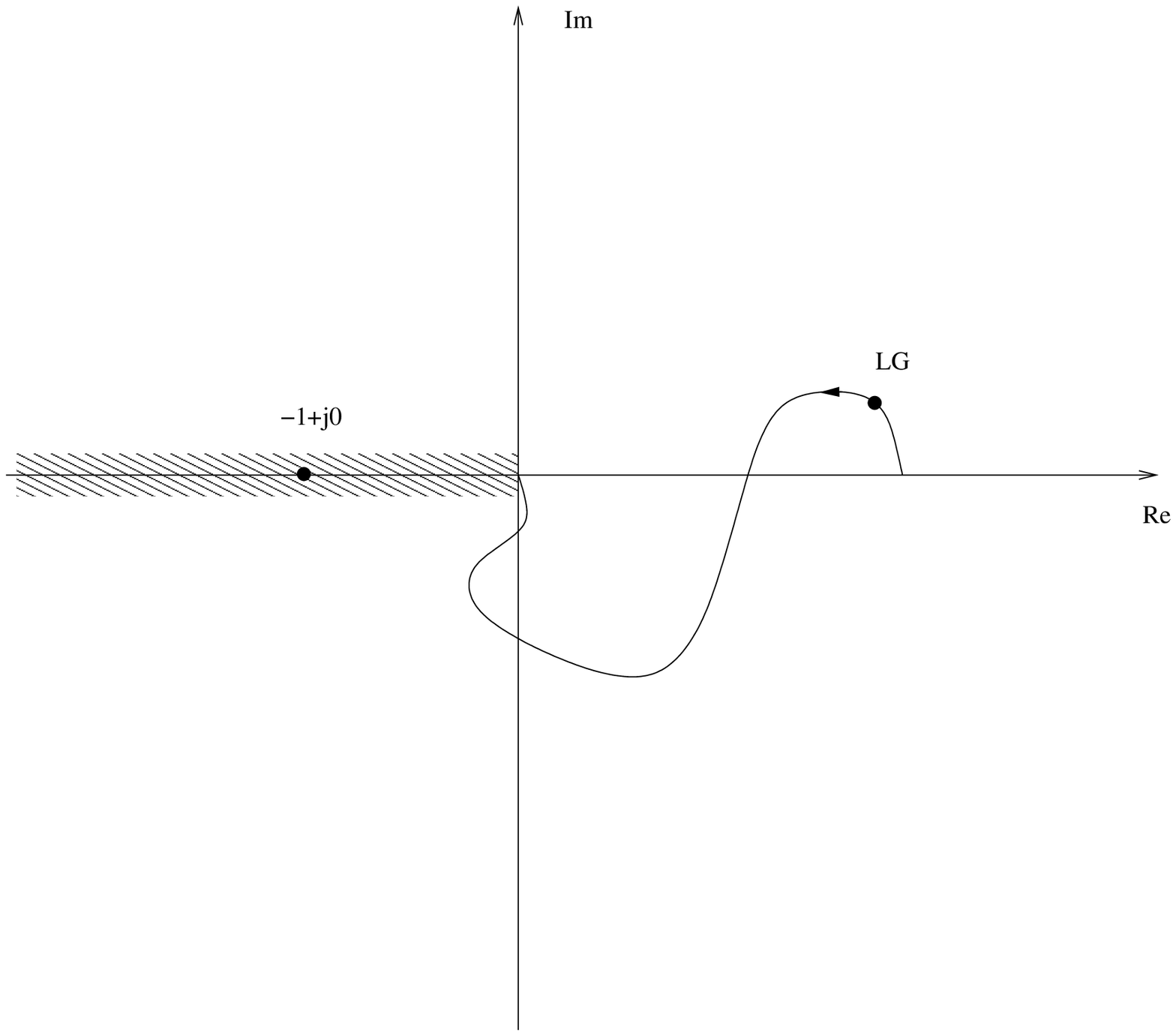}}
\caption{The passivity theorem. This plot shows two
  single-input single-output positive-real  transfer functions $M(s)$ and $N(s)$,
  both of whose Nyquist plots are  contained in  CRHP, and one of
  which
  is  contained in ORHP.
  Therefore, the Nyquist plot of the loop transfer function $M(s)N(s)$  cannot
  intersect the negative real axis.  Since the
  critical point $s=-1+\jmath 0$ cannot be encircled, it follows from
  the Nyquist stability criterion that  the negative-feedback
  interconnection of
  $M(s)$ and $N(s)$ must be internally  stable   }
\label{passivity}
\end{figure}

\newpage
\section{Sidebar 2\\ Applications to Electrical Circuits}

The  properties of a flexible structure with colocated actuators
and sensors  have counterparts in passive electrical circuits
driven by voltage or current sources.
Consider a resistor, inductor, capacitor  (RLC) electrical circuit with
$m$ voltage or current
sources. Assume that, for each voltage source input,
the  current flowing through the source is the corresponding output of the
system. Also, assume that,
for each current source input to the system, the  voltage across the source
is the corresponding  output of the system.   Let $v_1(t)$,
$\ldots$, $v_m(t)$ denote the voltage signals, and let $i_1(t)$,
$\ldots$, $i_m(t)$ denote the current signals. These signals are
dual in the sense that the product
$v_k(t)i_k(t)$ is equal to the power provided to the circuit
by the $k$th source at time $t$. Then,  let $u(t)$ be the vector of
 voltage- or current-source inputs at time $t$, and let $y(t)$ be the
vector of   voltage or current outputs at time $t$. Writing
\[
Y(s) = P(s) U(s),
\]
where
$P(s)$ is the transfer function matrix of the circuit, it follows that
the total power provided to the circuit by the
sources at time $t$ is given by $u^{\transpose}(t)y(t)$. As in the case of
a flexible structure
with colocated sensors and actuators, the
transfer function matrix
$P(s)$  is positive real.

Now suppose that each voltage source is connected in series with a
capacitor, and  that the corresponding system output is the
voltage across this capacitor divided by the capacitance. Also,
suppose that each current source is connected in parallel with an
inductor, and that the corresponding system output is the inductor
current divided by the inductance. This situation, which is
illustrated in Figure \ref{F1.0},  is
analogous to the case of a flexible structure with colocated force
actuation and position measurements since the current through a
capacitor is equal to the capacitance multiplied by the derivative of
the voltage across it. Also, the voltage across an inductor is
equal to the inductance multiplied by the derivative
of the current flowing through it.   Hence  each output variable is
such that its derivative is  a
 variable  that is dual to the corresponding source
 variable. Therefore, the transfer function matrix of the circuit $P(s)$   satisfies
 the
negative-imaginary condition
\[
\jmath(P(\jw) - P^{\transpose}(-\jw)) \geq 0.
\]

\newpage
\begin{figure}[H]
\centering
\mbox{\subfigure[Voltage source with series capacitor voltage
  measurement.]{
\psfrag{uk}{$u_k$}
\psfrag{yk}{$y_k$}
\includegraphics[width=8cm]{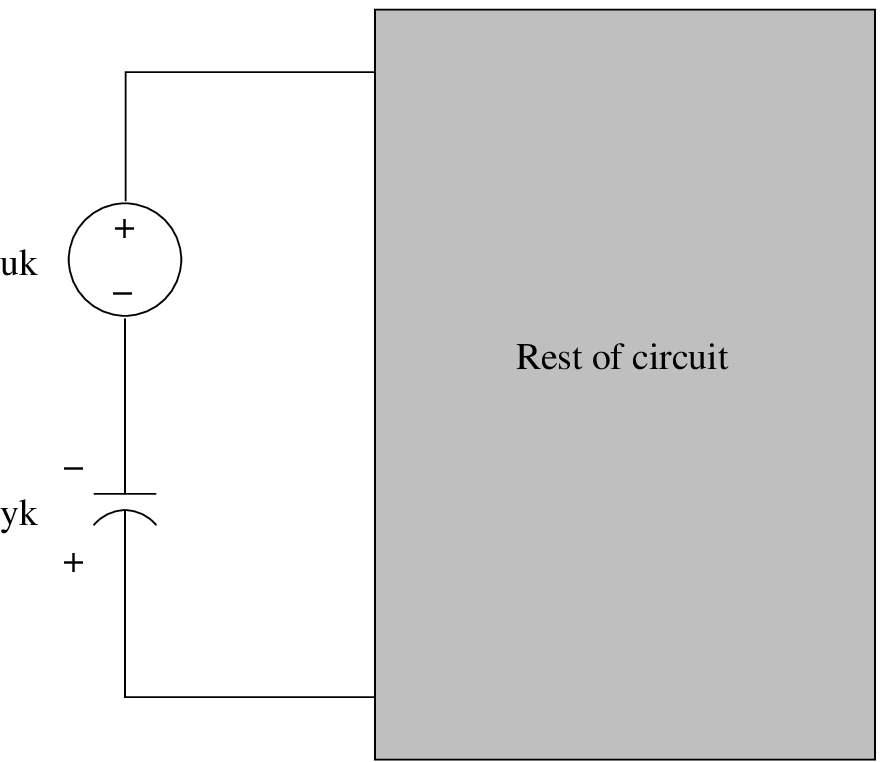}} \quad
\subfigure[Current source with parallel inductor current measurement.]{
\psfrag{uk}{$u_k$}
\psfrag{yk}{$y_k$}
\includegraphics[width=8cm]{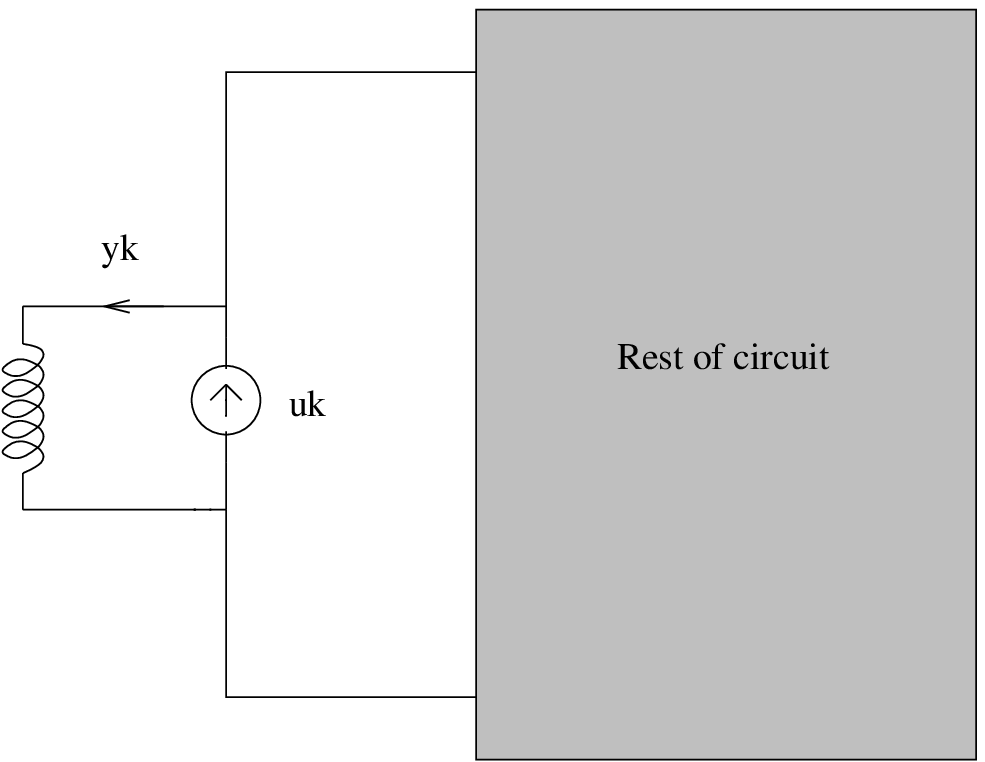}}
}
\caption{A resistor, inductor, capacitor (RLC) electrical
  circuit, where each input is a voltage  or
  current source. Also, each output corresponds to the voltage across a
  capacitor in series with a voltage source or the current through an
  inductor in parallel with a current source. This circuit is
  described by $P(s)$, the
   transfer function matrix from the vector of
  inputs to the vector of outputs. The transfer function matrix
  $P(s)$ is negative-imaginary.
  That is, the  transfer function matrix $P(s)$ has no poles in CRHP
  and satisfies the condition
$
\jmath(P(\jw) - P^{\transpose}(-\jw)) \geq 0
$
for all $\omega \geq 0$.
}
\label{F1.0}
\end{figure}

\newpage
\section{Sidebar 3\\ What Is Finsler's Theorem?}
Finsler's theorem, which is used in the proof of Lemma~\ref{IRC_SNI}, is summarized in the
  following lemma \cite{UHL79}.
\begin{lemma}
\label{Finsler}
  Let $M$ and $N$ be real symmetric  matrices  such that $M$ is
positive semidefinite and
$x^{\transpose}Nx \geq  0$
for all  real $x$ such that $M x = 0$.
Then there exists
$\bar \tau > 0$ such
that $N + \tau M \geq  0$ for all $\tau \geq \bar \tau$.
\end{lemma}

To illustrate Finsler's theorem,
let $M=\left[\begin{array}{cc} 1 & 0 \\ 0 & 0\end{array}\right]$ and
$N=\left[\begin{array}{cc} -1 & 0 \\ 0 & 1\end{array}\right]$.  All nonzero $x$ such that $M x
= 0$ are given by $x=\left[\begin{array}{c} 0 \\ \alpha\end{array}\right]$, where
$\alpha\in\mathbb{R}$ is nonzero.  Then $x^{\transpose}Nx =\alpha^2 > 0$. It
now follows from  Finsler's theorem
 that there exists $\bar\tau > 0$ such that $N + \tau M =
 \left[\begin{array}{cc}
    \tau -1 & 0 \\ 0 & 1\end{array}\right]\geq 0$ for all
$\tau\geq\bar\tau$. In this example,  $\bar\tau=1$.

\newpage
\section{Sidebar 4\\ How Are Rigid-Body Modes Handled?}
\renewcommand{\thetheorem}{S\arabic{theorem}}
\setcounter{theorem}{0}
\renewcommand{\thedefinition}{S\arabic{definition}}
\setcounter{definition}{0}
\renewcommand{\thefigure}{S\arabic{figure}}
\setcounter{figure}{0}

Output feedback control methods rely on output signal information measured through sensors to asymptotically stabilize all the internal states of a system.
In the case of a system that has unobservable modes that are not asymptotically stable, output feedback control cannot asymptotically stabilize the system.
Systems with rigid body modes,
which are characterized by a zero natural frequency, are an  example of systems that cannot be asymptotically stabilized by velocity feedback alone,
and position feedback is essential \cite[pp.~333--336]{M90}. Under velocity feedback alone, systems with rigid body modes can come to rest at a position other than
the origin of the state space.
The unobservability of the position states corresponding to the rigid body modes from the velocity outputs are the cause of this problem \cite[pp.~333--336]{M90}.

As a result of this problem,
rigid-body modes need special consideration in passivity approaches.  Typically a position feedback is applied before using
the passivity theorem, which is given in ``What Is Positive-real and Passivity Theory?''
Position feedback is applied
in an inner loop
before applying velocity feedback on the outer
loop. This  technique converts the rigid-body modes
into vibrational modes, which renders the corresponding position states
observable from the velocity outputs of the system.

Now consider positive-position control of systems with rigid-body modes.
The definitions of NI and SNI systems given in Definition \ref{D3} and
Definition \ref{D4}
require that NI and SNI systems have no poles at the origin. Hence,
theorems \ref{NIL} and \ref{thm:main_result} cannot directly handle
rigid-body modes. Although theorems \ref{NIL} and \ref{thm:main_result} cannot handle
rigid-body modes directly, a similar technique to the velocity feedback case involving a position feedback inner loop can also be used on NI systems that  have
rigid-body modes.
This position feedback inner loop is used to convert the rigid body modes into vibrational modes.  Then the result of \cite{XiPL1a}, which generalizes Theorem \ref{thm:main_result}   to allow for modes on the
imaginary axis except at the origin, can be applied  to guarantee internal stability of the overall feedback system.  Thus, the resulting inner feedback loop
consists of unity feedback and  proportional feedforward control
 to convert the rigid-body modes to vibrational modes. Then,
positive-position feedback is applied in the outer loop. An advantage in this case relative to  velocity feedback is that a
position sensor output is already available.

\newpage
\section{Author Information}

\noindent
Ian R. Petersen received the
Ph.D in electrical engineering in 1984 from the University of
Rochester. From 1983 to 1985 he was a postdoctoral fellow at the
Australian National University. In 1985 he joined the  University of
New South Wales at the
Australian Defence Force Academy where he is currently a Scientia
Professor and an Australian Research Council Federation Fellow in the
School of  Engineering and Information Technology. He was
Executive Director for Mathematics, Information and Communications for
the Australian Research Council from 2002 until 2004, and acting
Deputy Vice-Chancellor Research for the University of New South Wales
in 2004 and 2005. He has
served as an associate editor for the {\em IEEE Transactions on Automatic
Control}, {\em Systems and Control Letters}, {\em Automatica}, and
{\em SIAM Journal on
Control and Optimization}.  He is currently
 {\em Automatica} editor for control and estimation theory. He is a Fellow of the IEEE.  His
research interests are in robust and nonlinear control theory, quantum control
theory, and stochastic control theory.

 Contact details:
 School of Information Technology and Electrical Engineering,
  University of New South Wales at the Australian Defence Force
  Academy,      Canberra ACT 2600, Australia,
email:     i.petersen@adfa.edu.au, telephone +61 2 62688446, fax +61 2 62688443.

\noindent Alexander Lanzon received the B.Eng.(Hons). degree in electrical engineering
from the University of Malta in 1995, and the Masters' and Ph.D. degrees in control engineering
from the University of Cambridge in 1997 and 2000, respectively. Before joining the
University of Manchester in 2006, he held academic positions at Georgia Institute of Technology and the Australian National University.
He received earlier research training at Bauman Moscow State Technical University and
industrial training at ST-Microelectronics Ltd., Yaskawa Denki Tokyo Ltd., and
National ICT Australia Ltd. He is a Fellow of the IET, a Senior Member of the IEEE,
and a member of AIAA. His research interests include the fundamental theory of
feedback control systems, robust control theory, and its applications to
aerospace control, in particular, control of new UAV concepts.

Contact details:
Control Systems Centre, School of Electrical and Electronic Engineering,
University of Manchester, Sackville Street, Manchester M13 9PL, UK,
Email: a.lanzon@ieee.org, telephone: +44-161-306-8722, fax: +44-161-306-8722.
\end{document}